\documentclass[11pt]{article}
\usepackage{fullpage}
\usepackage[small,compact]{titlesec}
\usepackage[linesnumbered,algoruled,titlenotnumbered,vlined]{algorithm2e}
\usepackage[T1]{fontenc}
\usepackage{amsmath}
\usepackage{amssymb}
\usepackage{mathtools}
\usepackage{amsthm}
\usepackage{hyperref}
\usepackage{xspace}
\usepackage{enumerate}
\usepackage{multirow}
\usepackage{tablefootnote}
\usepackage{color}
\usepackage{paralist}
	
\SetKwFor{InfLoop}{repeat}{}{endrepeat}
\SetKw{lOr}{or}
\SetKw{lAnd}{and}
\DontPrintSemicolon
 
\newtheorem{theorem}{Theorem}[section]
\newtheorem{lemma}[theorem]{Lemma}

\theoremstyle{definition}

\renewenvironment{proof}{\vspace{-0.05in}\noindent{\bf Proof:}}%
       {\hspace*{\fill}$\Box$\par}
        {\hspace*{\fill}$\Box$\par}
        {\hspace*{\fill}$\Box$\par}

\newcommand{\eps}{\epsilon}
\newcommand{\E}{\mathbb{E}}

\newcommand{\R}{{\mathbb{R}}}

\def\script#1{\mathcal{#1}}
\def\card#1{|#1|}
\def\set#1{\left\{#1\right\}}

\def\etal{{\em et al.}\@\xspace}

\DeclareMathOperator{\Ex}{\mathbf{E}}

\DeclareMathOperator{\poly}{poly}

\DeclareMathOperator{\argmax}{argmax}

\def\bx{\mathbf{x}}

\def\by{\mathbf{y}}

\def\bz{\mathbf{z}}
\def\one{\mathbf{1}}
\def\zero{\mathbf{0}}
\def\bp{\mathbf{p}}

\def\bb{\mathbf{b}}

\def\hC{\widehat{C}}

\def\mypar#1{\medskip\noindent \textbf{#1}}

\def\opt{\ensuremath{\mathrm{OPT}}}

\def\D{\script{D}}
\def\I{\script{I}}

\def\Greedy{\ensuremath{\mathsf{Greedy}}\xspace}
\def\DTreshGreedy{\ensuremath{\mathsf{DThreshGreedy}}\xspace}
\def\DCGreedy{\ensuremath{\mathsf{DCGreedy}}\xspace}
\def\DCGreedyRel{\ensuremath{\mathsf{DCGreedyRel}}\xspace}
\def\DCGreedySol{\ensuremath{\mathsf{DCGreedySol}}\xspace}
\def\GreedyStep{\ensuremath{\mathsf{GreedyStep}}\xspace}

\def\Alg{\ensuremath{\mathsf{Alg}}\xspace}
\def\ParallelAlg{\ensuremath{\mathsf{ParallelAlg}}\xspace}
\def\AlgSol{\ensuremath{\mathsf{AlgSol}}\xspace}
\def\AlgRel{\ensuremath{\mathsf{AlgRel}}\xspace}
\def\Sbest{\ensuremath{S_{\mathrm{best}}}}

\renewcommand{\colon}{\nobreak\mskip2mu\mathpunct{}\nonscript
  \mkern-\thinmuskip{:}\mskip6muplus1mu\relax}

\newenvironment{intro-theorem}[2][]{\begingroup\renewcommand{\thetheorem}{\ref{#2}}
\begin{theorem}[#1]}
{\end{theorem}
\addtocounter{theorem}{-1}
\endgroup
}

\newenvironment{appendix-theorem}[1]{\begingroup\renewcommand{\thetheorem}{\ref{#1}}
\begin{theorem}}
{\end{theorem}
\addtocounter{theorem}{-1}
\endgroup
}

\newenvironment{appendix-lemma}[1]{\begingroup\renewcommand{\thetheorem}{\ref{#1}}
\begin{lemma}}
{\end{lemma}
\addtocounter{theorem}{-1}
\endgroup
}

\title{A New Framework for Distributed Submodular Maximization}
\author{
Rafael da Ponte Barbosa\thanks{Department of Computer Science and DIMAP, University of Warwick. {\tt rafael@dcs.warwick.ac.uk}.}
\and
Alina Ene\thanks{Department of Computer Science, Boston University. {\tt a.ene@bu.edu}.}
\and
Huy L. Nguy\~{\^{e}}n\thanks{College of Computer and Information Science, Northeastern University. {\tt hlnguyen@cs.princeton.edu}.}
\and
Justin Ward\thanks{School of Computer and Communication Sciences, EPFL, Lausanne, Switzerland. {\tt justin.ward@epfl.ch}.}
}

\begin{document}
\maketitle


\begin{abstract}
A wide variety of problems in machine learning, including exemplar clustering, document summarization, and sensor placement, can be cast as constrained submodular maximization problems. A lot of recent effort has been devoted to developing distributed algorithms for these problems. However, these results suffer from high number of rounds, suboptimal approximation ratios, or both. We develop a framework for bringing existing algorithms in the sequential setting to the distributed setting, achieving near optimal approximation ratios for many settings in only a constant number of MapReduce rounds. Our techniques also give a fast sequential algorithm for non-monotone maximization subject to a matroid constraint.
\end{abstract}

\section{Introduction}

The general problem of maximizing a submodular function appears in a variety of contexts, both in theory and practice.  From a theoretical perspective, the class of submodular functions is extremely rich, including examples as varied as cut functions of graphs and digraphs, the Shannon entropy function, weighted coverage functions, and log-determinants.
Recently, there has been a great deal of interest in practical applications of submodular optimization, as well.  Variants of facility location, sampling, sensor selection, clustering, influence maximization in social networks, and welfare maximization problems are all instances of submodular maximization.  In practice, many of these applications involve processing enormous datasets requiring \emph{efficient, distributed algorithms}.

In contrast, most successful approaches for submodular maximization have been based on \emph{sequential} greedy algorithms, including the standard greedy algorithm \cite{Nemhauser1978a,Fisher1978}, the continuous greedy algorithm \cite{Calinescu2011,Feldman2011a}, and the double greedy algorithm \cite{Buchbinder2012}.   Indeed, such approaches attain the best-possible, tight approximation guarantees in a variety of settings \cite{Feige1998,Vondrak2009,Dobzinski2012}, but unfortunately they all share a common limitation, inherited from the standard greedy algorithm: they are inherently sequential. This presents a seemingly fundamental barrier to obtaining efficient, highly parallel variants of these algorithms.

\subsection{Our Contributions}
\label{sec:our-results}

As demonstrated by the extensive prior works on submodular maximization, the community has a good understanding of the problem under remarkably general types of constraints, which are handled by a small collection of general algorithms. In contrast, the existing works in the distributed setting are either tailored to special cases, giving approximation factors far from optimal or requiring a large number of distributed rounds.  One cannot help but wonder if, instead of retracing the individual advances made in the sequential setting over the last few decades, it may be possible to obtain a generic technique to carry over the algorithms in the sequential setting to the parallel world.

In this work, we present a significant step toward resolving the above question. Our main contribution is a \emph{generic parallel algorithm} that allows us to parallelize a broad class of sequential algorithm with almost no loss in performance. The crux of our approach is a \emph{common abstraction} that allows us to capture and parallelize both the standard and continuous greedy algorithms, and it provides a novel unifying perspective for these algorithmic paradigms. Our framework leads to the first distributed algorithms that nearly match the state of the art approximation guarantees for the sequential setting in only a constant number of rounds. In the following, we summarize our main contributions.

\mypar{A parallel greedy algorithm.}
We obtain the following general result by parallelizing the standard greedy algorithm:
\begin{intro-theorem}{thm:parallel-greedy-main}
Let $f : 2^V \to \R_+$ be a submodular function, and $\I \subseteq 2^V$ be a hereditary set system\footnote{A set system is hereditary if for any $S \in \I$, all subsets of $S$ are also in $\I$.}. For any $\epsilon > 0$ there is a randomized distributed  $O(1/\epsilon)$-round  algorithm that can be implemented in the MapReduce framework\footnote{We define the MapReduce model in Section~\ref{sec:model}.}. The algorithm is an $(\alpha - O(\eps))$-approximation with constant probability for the problem $\max_{S \in \I}f(S)$, where $\alpha$ is the approximation ratio of  the standard, sequential greedy algorithm for the same problem.
\end{intro-theorem}
Our constant number of rounds is a significant improvement over the sample and prune technique of \cite{Kumar2013}, which requires a number of rounds depending logarithmically on the \emph{value} of the single best element.  Remarkably, even for the especially simple case of a cardinality constraint, no previous work could get close to the approximation ratio of the simple sequential greedy algorithm in a constant number of rounds. Our framework nearly matches the approximation ratio of greedy in all situations in a constant number of rounds and immediately resolves this problem.
\begin{table}[t!]
\let\oldarraystretch=\arraystretch
\renewcommand{\arraystretch}{1.23}
\setlength{\tabcolsep}{3.75pt}
\centering
\begin{tabular}{|c|c|c|c|}
\hline
\multicolumn{4}{|c|}{{\bf Monotone functions}}\\ 
\hline
{\small Constraint} & {\small Rounds} & {\small Approx.} & {\small Citation}\\ \hline
\multirow{3}{*}{{\small cardinality}} & $O\big({\log \Delta \over \eps} \big)$ & $1-\frac{1}{e}-\eps$ &\cite{Kumar2013}\\ \cline{2-4}
& 2 & 0.545 &\cite{Mirrokni2015}\\ \cline{2-4}
& $O\big({1 \over \eps} \big)$ & $1 - \frac{1}{e}-\eps$ & {\small Theorem \ref{thm:parallel-greedy-main}}\\ \hline
\multirow{3}{*}{{\small matroid}} & $O\big({\log \Delta \over \eps} \big)$ & ${1 \over 2} - \eps$ & \cite{Kumar2013}\\ \cline{2-4}
& $2$ & $\frac{1}{4}$ & \cite{Barbosa2015}\\ \cline{2-4}
& $O\big({1 \over \eps} \big)$ & $1-\frac{1}{e}-\eps$ & {\small Theorem \ref{thm:parallel-cg-main}}\\ \hline
\multirow{3}{*}{{\small $p$-system}} & $O\big({\log \Delta \over \eps} \big)$ & $\frac{1}{p+1} -\eps$ & \cite{Kumar2013}\\ \cline{2-4}
& $2$ & $\frac{1}{2(p+1)}$ & \cite{Barbosa2015}\\ \cline{2-4}
& $O\big({1 \over \eps} \big)$ & $\frac{1}{p+1} -\eps$ & {\small Theorem \ref{thm:parallel-greedy-main}} \\ \hline
\end{tabular}
\renewcommand{\arraystretch}{1.5}
\setlength{\tabcolsep}{3.5pt}
\begin{tabular}{|c|c|c|c|}
\hline
\multicolumn{4}{|c|}{{\bf Non-monotone functions}}\\
\hline
{\small Constraint} & {\small Rounds} & {\small Approx.} & {\small Citation}\\
\hline
\multirow{2}{*}{{\small cardinality}} & 2 & ${1 - {1 \over m} \over  2 + e }$ & \cite{Mirrokni2015}\\ \cline{2-4}
& 2 & $(1 - {1 \over m}) {1 \over e}(1 - {1 \over e})$ & {\small Theorem~\ref{thm:two-round-cardinality}}  \\ \hline
\multirow{3}{*}{{\small matroid}} & 2 & ${1 \over 10}$ & \cite{Barbosa2015} \\ \cline{2-4}
& 2 & ${1 - {1 \over m} \over  2 + e }$ & {\small Theorem \ref{thm:two-round}} \\ \cline{2-4}
& $O\big( {1 \over \eps} \big)$ & ${1 \over e} - \eps$ & {\small Theorem \ref{thm:parallel-cg-main}} \\ \hline 
\multirow{2}{*}{{\small $p$-system}} & 2 & ${1 \over 2 + 4(p + 1)}$ & \cite{Barbosa2015} \\ \cline{2-4}
& 2 & ${3(1 - {1 \over m}) \over 5p + 7 + {2 \over p}}$ & {\small Theorem \ref{thm:two-round}}\\ \hline
\end{tabular}
\let\arraystretch=\oldarraystretch
\caption{New results for distributed submodular maximization. Here $\Delta = \max_{i \in V} f(\{i\})$ and $m$ is the number of machines. In the results of \cite{Kumar2013}, in the number of rounds, $\Delta$ can be replaced by the maximum size of a solution. All algorithms in previous works and ours are randomized and the approximation guarantees stated hold in expectation, and they can be strengthened to hold with high probability by repeating the algorithms in parallel.}
\label{tab:results}
\end{table}

\mypar{A parallel continuous greedy algorithm.}
We obtain new distributed approximation results for maximization over matroids, by using a heavily discretized variant of the measured continuous greedy algorithm, obtaining approximation guarantees nearly matching those attained by the continuous greedy in the sequential setting.
\begin{intro-theorem}{thm:parallel-cg-main}
Let $f : 2^V \to \R_+$ be a submodular function, and $\I \subseteq 2^V$ be a matroid.  For any $\epsilon > 0$ there is a randomized distributed  $O(1/\epsilon)$-round  algorithm that can be implemented in the MapReduce framework.  The algorithm is an $(\alpha - O(\eps))$-approximation with constant probability for the problem $\max_{S \in \I}f(S)$, where $\alpha$ is $(1 - 1/e)$ for monotone $f$ and $1/e$ for general $f$.
\end{intro-theorem}

\mypar{Improved two-round algorithms and fast sequential algorithms.}
We also give improved two-round approximations for non-monotone submodular maximization under hereditary constraints.  We make use of the same ``strong greedy property'' utilized in \cite{Barbosa2015} but attain approximation guarantees strictly better than were given there.  Our algorithm is based on a combination of the standard greedy algorithm \Greedy and an additional, arbitrary algorithm \Alg.  Again, we suppose that $f$ is a (not necessarily monotone) submodular function and $\I$ is any hereditary constraint.  In the following theorems and throughout the paper, $n \coloneqq |V|$ is the size of the ground set, $k \coloneqq \max_{S \in \script{I}} |S| $ is the maximum size of a solution, and $m$ is the number of machines employed by the distributed algorithm.
\begin{intro-theorem}{thm:two-round}
Suppose that \Greedy satisfies the strong greedy property with constant $\gamma$ and that \Alg is a $\beta$-approximation for the problem $\max_{S \in \I}f(S)$. Then there is a randomized, two-round distributed algorithm that achieves a $( 1 - {1 \over m}){\beta \gamma  \over \beta + \gamma}$ approximation in expectation for $\max_{S \in \I}f(S)$.
\end{intro-theorem}
We show that by simulating the machines in this last distributed algorithm, we also obtain a fast, sequential algorithm for maximizing a non-monotone submodular function subject to a matroid constraint.  Our algorithm shows that one can preprocess the instance in $O({n \over \eps} \log{n})$ time and obtain a set $X$ of size $O(k/\eps)$ so that \emph{it suffices to solve the problem on $X$}.  By using a variant of the continuous greedy algorithm on the resulting set $X$, we obtain the following result.

\begin{intro-theorem}{thm:fast-matroid}
There is a sequential, randomized $(\frac{1}{2+e} -\eps)$-approximation algorithm for the problem $\max_{S \in \I}f(S)$, where $\I$ is any matroid constraint, running in time $O({n \over \eps} \log{n}) + \poly({k \over \eps})$.
\end{intro-theorem}

As a final application of our techniques, we obtain a very simple two-round distributed algorithm for monotone maximization subject to a cardinality constraint.

\begin{intro-theorem}{thm:half-cardinality}
There is a randomized, two-round, distributed algorithm achieving a ${1 \over 2} - \eps$ approximation in expectation for $\max_{S : |S| \le k}f(S)$, where $f$ is a monotone~function.
\end{intro-theorem}

\subsection{Techniques}

In contrast with the previous framework by~\cite{Kumar2013} which is based on repeatedly eliminating bad elements, our framework is more in line with the greedy approach of identifying good elements. The algorithm maintains a pool of good elements that is grown over several rounds. In each round, the elements are partitioned randomly into groups. Each group selects the best among its elements and the good pool using the sequential algorithm. Finally, the best elements from all groups are added to the good pool. The best solution among the ones found in the execution of the algorithm is returned at the end. The previous works based on 2 rounds of MapReduce such as~\cite{Barbosa2015} can be viewed as a single phase of our algorithm.  The first phase can already identify a constant fraction of the weight of the solution, thus obtaining a constant factor approximation. However, it is not clear how to obtain the best approximation factor from such an approach. Our main insight is that, with a right measure of progress, we can grow the solution iteratively and obtain solutions that are arbitrarily close to those of sequential algorithms.  We show that after only $O(\frac{1}{\epsilon})$ rounds, the pool of good elements already contains a good solution with constant probability.

\subsection{Related Work}
\label{sec:related-work}
There has been a recent push toward obtaining fast, practical algorithms for submodular maximization problems arising in a variety of applied settings.  Research in this direction has yielded a variety of techniques for speeding up the continuous greedy algorithm for monotone maximization~\cite{Badanidiyuru2014,Mirzasoleiman2015}, as well as new approaches for non-monotone maximization based on insights from both the continuous greedy and double greedy algorithms~\cite{Buchbinder2014a,Buchbinder2014}.  Of particular relevance to our results is the case of maximization under a matroid constraint.  Here, for monotone functions the fastest current sequential algorithm gives a $1 - 1/e - \epsilon$ approximation using $O(\frac{\sqrt{k} n}{\epsilon^5}\ln^2(\frac{n}{\epsilon}) + \frac{k^2}{\epsilon})$ value queries. For non-monotone functions, Buchbinder \etal~\cite{Buchbinder2014} give an $\frac{1+e^{-2}}{4} > 0.283$-approximation in time $O(kn\log n + Mk)$, where $M$ is the time required to compute a perfect matching on bipartite graph with $k$ vertices per side.  They also give a simple, combinatorial $1/4$-approximation in time $O(kn\log n)$.  In comparison, the sequential algorithm we present here is faster by a factor of $\Omega(k)$, at the cost of a slightly-weaker $\frac{1}{2+e} > 0.211$-approximation.

Work on parallel and distributed algorithms for submodular maximization has been comparatively limited.  Early results considered the special case of maximum $k$-coverage, and attained an $O(1 - 1/e - \epsilon)$-approximation \cite{Chierichetti2010,Blelloch2011}. Later, Kumar \etal \cite{Kumar2013} considered the more general problem of maximizing an arbitrary monotone submodular function subject to a matroid, knapsack, or $p$-system constraint.  Their approach attains a $\frac{1}{2+\epsilon}$ approximation for matroids, and requires $O(\frac{1}{\epsilon}\log \Delta)$ MapReduce rounds, where $\Delta$ is the value of the best single element.  More generally, they obtain a $\frac{1}{p+1+\epsilon}$ approximation for $p$-systems in $O(\frac{1}{\epsilon}\log \Delta)$ rounds. The factor of $\log \Delta$ in the number of rounds is inherent in their approach: they adapt the threshold greedy algorithm, which sequentially picks elements in $\log \Delta$ different thresholds. In another line of work, Mirzasoleiman \etal\cite{Mirzasoleiman2013} introduced a simple, two-round distributed greedy algorithm for submodular maximization.  While their algorithm is only an $O(\frac{1}{m})$-approximation in the worst case, it performs very well in practice, and attains provable constant-factor guarantees for submodular functions exhibiting certain additional structure.  Barbosa \etal \cite{Barbosa2015} recently gave a more sophisticated analysis of this approach and showed that, if the initial distribution of elements is performed randomly, the algorithm indeed gives an expected, constant-factor guarantee for a variety of problems. Finally, Mirrokni and Zadimoghaddam~\cite{Mirrokni2015} gave the currently-best $0.545$-approximation for the cardinality constraint case using only 2 rounds of MapReduce.

\section{The model}
\label{sec:model}
We adopt the most stringent MapReduce-style model among~\cite{KSV10,GSZ11,BKS13,ANOY14}, the Massively Parallel Communication (MPC) model from~\cite{BKS13} as specified by \cite{ANOY14}. Let $N$ be the size of the input. In this model, there are $M$ machines each with space $S$. The total memory of the system is $M\cdot S = O(N)$, which is at most a constant factor more than the input size. Computation proceeds in synchronous rounds. In each round, each machine can perform local computation and at the end, it can send at most a total of $O(S)$ words to other machines. These $O(S)$ words could form a single message of size $S$, $S$ messages of size $1$, or any other combination whose sum is at most $O(S)$. Following~\cite{KSV10}, we restrict both $M, S < N^{1-\Omega(1)}$. The typical main complexity measure is the number of rounds.

Note that not all previous works on MapReduce-style algorithms for submodular maximization satisfy the strict requirements of the MPC model. For instance, as stated, the previous work by Kumar et al.~\cite{Kumar2013} uses $\Theta(N\log N)$ total memory and thus it does not fit in this model (though it might be possible to modify their algorithms to satisfy this).

We assume that the size of the solution is at most $N^{1-2c}$ for some constant $0 < c < 1/2$. Thus, an entire solution can be stored on a single machine in the model. This assumption is also used in previous work such as~\cite{Mirrokni2015}.

\section{Preliminaries}
\label{sec:prelim}

A function $f : 2^V \to \R_+$ is submodular if and only if $f(A \cup \{e\}) - f(A) \ge f(B\cup \{e\}) - f(B)$ for all $A \subseteq B$ and $e \not\in B$.  If $f(A \cup \{e\}) - f(A) \ge 0$ for all $A$ and $e \not\in A$ we say that $f$ is \emph{monotone}.  Here we consider the general problem $\max \{f(S) : S \subseteq V, S \in \I\}$, where $\I$ is any hereditary constraint (i.e., a downward-closed family of subsets of $V$).

Throughout the paper, $n \coloneqq |V|$ is the size of the ground set, $k \coloneqq \max_{S \in \script{I}} |S| $ is the maximum size of a solution, and $m$ is the number of machines employed by the distributed algorithm.

We shall consider both monotone and non-monotone submodular functions. However, the following simple observation shows that even non-monotone submodular functions are monotone when restricted to the optimal solution of a problem of the sort we consider.

\begin{lemma}\label{thm:monotonicity-opt}
Let $f$ be a submodular function and $\opt = \arg\max_{S \in \I} f(S)$ for some hereditary constraint $\I$. Then, $f(A \cap \opt) \le f(B \cap \opt)$ for all $A \subseteq B$.
\end{lemma}
\begin{proof}
Consider $X \subseteq \opt$ and $e\in \opt\setminus X$. By submodularity, $f(X\cup\{e\}) - f(X) \ge f(\opt) - f(\opt\setminus\{e\})$. On the other hand, because $\mathcal{I}$ is hereditary, $\opt\setminus\{e\}$ is feasible and thus $f(\opt) \ge f(\opt\setminus\{e\})$. Therefore $f(X\cup\{e\}) - f(X)\ge 0$ for all $X$ and $e \in \opt \setminus X$.
\end{proof}

\mypar{Continuous extensions.}
In this paper, we work with two standard continuous extensions of submodular functions, the multilinear extension and the Lov\'asz extension. The \emph{multilinear extension} of $f$ is the function $F: [0, 1]^V \to \R_+$ such that $F(\bx) = \Ex[f(R(\bx))]$, where $R(\bx)$ is a random subset of $V$ in which each element $e$ appears independently with probability $x_e$. 

The \emph{Lov\'asz extension} of $f$ is the function $f^- : [0,1]^V \to \R_+$ such that $f^-(\bx) = \Ex_{\theta \in \script{U}(0,1)}[f(\{e : x_e \ge \theta\})]$, where $\script{U}(0,1)$ is the uniform distribution on [0,1].
For any submodular function $f$, the Lov\'asz extension $f^-$ satisfies: $f^-(\one_{S}) = f(S)$ for all $S \subseteq V$; $f^-$ is convex; and the restricted scale invariance property $f^-(c \cdot \bx) \ge c \cdot f^-(\bx)$ for any $c \in [0,1]$.  We shall make use of the following lemmas.
\begin{lemma}[\cite{Barbosa2015}, Lemma 1]
Let $S$ be a random set with $\Ex[\one_{S}] = c \cdot \bp$ (for $c \in [0,1]$).  Then, $\Ex[f(S)] \ge c \cdot f^-(\bp)$.
\label{lem:lovasz}
\end{lemma}

\begin{lemma}\label{lem:lovasz-1}
Let $f : 2^V \to \R_+$ be a submodular function that is monotone when restricted to $X \subseteq V$.  Further, let $T, S \subseteq X$, and let $R$ be a random subset of $T$ in which every element occurs with probability at least $p$.  Then, $\Ex[f(R \cup S)] \ge p\cdot f(T \cup S) + (1 - p)f(S)$.
\end{lemma}
\begin{proof}
Recall that $f^-$ is the  Lov\'{a}sz extension of $f$. Since $f^-$ is convex,
\[
\Ex[f(R \cup S)] = \Ex[f^-(\one_{R \cup S})] \ge f^-(\Ex[\one_{R \cup S}]) =
f^-(\Ex[\one_{R\setminus S}] + \one_S).
\]
Since every element of $T$ occurs in $R$ with probability at least $p$, we have $\Ex[\one_{R \setminus S}] \ge p \cdot \one_{T \setminus S}$.  Then, since $f$ is monotone with respect to $X \supseteq S \cup T$, we must have:
\[ f^-(\Ex[\one_{R\setminus S}] + \one_S) \ge
f^-(p \cdot \one_{T \setminus S} + \one_S). \]
Finally, from the definition of $f^-$, we have
\[ f^-(p \cdot \one_{T \setminus S} + \one_S) = p \cdot f(T \cup S) + (1 - p)f(S).\]
\end{proof}

\section{Generic Parallel Algorithm for Submodular Maximization}
\label{sec:parallel-framework}

In this section, we give a generic approach for parallelizing \emph{any sequential algorithm $\Alg$} for the problem $\max_{S \subseteq V \colon S \in \I} f(S)$, where $f: 2^V \rightarrow \R_+$ is a submodular function and $\I \subseteq 2^V$ is a hereditary constraint.

As a starting point, we need a common abstract description of existing sequential algorithms. Towards that end, we turn to the standard Greedy and Continuous Greedy algorithms for inspiration. The Greedy algorithm directly constructs a solution, whereas the Continuous Greedy algorithm first constructs a fractional solution $\bx$ which is then rounded to get an integral solution. In the common abstraction, we will need both the integral solution \emph{and} the support of the fractional solution $\bx$. To account for this, we will have the algorithm $\Alg$ return a pair of sets, $(\AlgSol(V), \AlgRel(V))$, where $\AlgSol(V) \in I$ is a feasible solution for the problem and $\AlgRel(V)$ is a set providing additional information. When using the standard Greedy algorithm for $\Alg$, $\AlgSol(V)$ and $\AlgRel(V)$ will both be equal to the Greedy solution. When using the Continuous Greedy algorithm for $\Alg$, $\AlgSol(V)$ will be the integral solution and $\AlgRel(V)$ will be the support of the fractional solution constructed by the Continuous Greedy algorithm.

More importantly, we will need an abstraction that captures the greedy behavior of these algorithms. We encapsulate the crucial properties of greedy-like algorithms in the following definition. We believe that this framework is one of the most valuable and insightful contributions of this work, and it provides a general abstraction for a broader class of algorithms.

We assume that the algorithm $\Alg$ satisfies the following properties.

\begin{enumerate}
\item ($\alpha$-Approximation) For every input $N \subseteq V$, $\AlgSol(N)$ is an $\alpha$-approximate solution to $\max_{S \subseteq N \colon S \in \I} f(S)$.   

\item (Consistency) Let $A$ and $B$ be two disjoint subsets of $V$. Suppose that, for each element $e \in B$, we have $\AlgRel(A \cup \set{e}) = \AlgRel(A)$. Then $\AlgSol(A \cup B) = \AlgSol(A)$.
\end{enumerate}

Armed with this definition, we can now describe our approach for parallelizing an abstract sequential algorithm $\Alg$ with almost no loss in the approximation guarantee.

\mypar{Parallel algorithm $\ParallelAlg$ based on $\Alg$.}
As before, let $\alpha$ be the approximation guarantee of the sequential algorithm $\Alg$. Let $s \coloneqq \max_{N \subseteq V} |\AlgSol(N) \cup \AlgRel(N)|$ be the maximum size of the sets returned by $\Alg$. Let $\eps > 0$ be the desired accuracy, i.e., we will aim that $\ParallelAlg$ achieves an $(\alpha - \eps)$ approximation.

The algorithm uses $g \coloneqq \Theta(1/(\alpha \eps))$ groups of machines with $m$ machines in each group (and thus the total number of machines is $gm$). The number $m$ of machines can be chosen arbitrarily and it will determine the amount of space needed on each machine, since the dataset is divided roughly equally among each of the $m$ machines in each group. An optimal setting is $gm \coloneqq O(\sqrt{n/s})$.

The algorithms performs $\Theta(1/\eps)$ \emph{runs}. Throughout the process, we maintain two quantities: an \emph{incumbent} solution $\Sbest$, which is the best solution produced on any single machine so far in the process, and a \emph{pool} of elements $C \subseteq V$ (we assume that the incumbent solution is stored on one designated machine).

Each run of the algorithm proceeds as follows. Amongst each group of $m$ machines, we partition $V$ uniformly at random; each element $e$ chooses an index $i \in [m]$ uniformly and independently at random and is assigned to the $i$th machine in the group.  We do this separately for each group of machines, i.e., each element appears on exactly one machine in each group.  For an individual machine $i \in [gm]$, let $X_{i,r}$ denote the set of elements that are assigned to $i$ in run $r$ by this procedure.  Additionally, we place on each machine the same pool of elements $C_{r-1}$, constructed at the end of run $r-1$.

Once the elements have been distributed as described above, on each machine $i$, we run the algorithm $\Alg$ on the input $X_{i, r} \cup C_{r - 1}$ on the machine to obtain $(\AlgSol(X_{i, r} \cup C_{r - 1}), \AlgRel(X_{i, r} \cup C_{r - 1}))$. We update the incumbent solution $\Sbest$ to be the better of the current solution $\Sbest$ and the solutions $\AlgSol(X_{i, r} \cup C_{r - 1})$ constructed on each of the machines; this is achieved by having each machine send $\AlgSol(X_{i, r} \cup C_{r - 1})$ to some designated machine maintaining $\Sbest$, and this machine will update $\Sbest$ in the next round. We update the pool by setting $C_r \coloneqq C_{r - 1} \bigcup_{i} \AlgRel(X_{i, r} \cup C_{r - 1})$; this is achieved by having each machine send $\AlgRel(X_{i, r} \cup C_{r - 1})$ to every other machine, and thus ensuring that the pool $C_r$ is available on each machine during the next round.

At the end of the $\Theta(1/\eps)$ runs, the algorithm returns the incumbent solution $\Sbest$. This completes the description of our algorithm.

\mypar{Avoiding duplicating the dataset.}
The algorithm above partitions the dataset over $\Theta(1/\eps)$ groups of machines and thus it duplicates the dataset $\Theta(1/\eps)$ times (this problem also applies to previous work~\cite{Mirrokni2015}). This is done in order to achieve the best theoretical guarantee on the number of runs, but in practice it is undesirable to duplicate the data. Instead, we can use a single group of $m$ machines and perform the computation of a single run sequentially over $\Theta(1/\eps)$ sub-run, where each sub-run performs the computation of one of the group of machines. This will lead to an algorithm that performs $\Theta(1/\eps^2)$ runs using $m$ machines and it does not duplicate the dataset. 

\mypar{The analysis.}
We devote the rest of this section to the analysis of the algorithm $\ParallelAlg$. We start by noting that, if we choose $g$ and $m$ so that $gm = O(\sqrt{n/s})$, the algorithm uses the following resources and thus it satisfies the requirements of the model in Section~\ref{sec:model}.
\begin{lemma}
	$\ParallelAlg$ can be implemented in the parallel model in Section~\ref{sec:model} using the following resources.
	\begin{compactitem}
		\item The number of rounds is $O(1/\eps)$.
		\item The number of machines is $O(\sqrt{n/s})$.
		\item The amount of space used on each machine is $O(\sqrt{n s} / (\eps \alpha))$ with high probability.
		\item In each round, the total amount of communication from a machine to all other machines is $O(\sqrt{n s}/(\eps \alpha))$ with high probability. The total amount of communication over all machines in a given round is $O(n/(\eps\alpha))$.
	\end{compactitem}
\end{lemma}
\begin{proof}
	We will choose $gm \coloneqq \sqrt{n/s}$ as our number of machines. Using this choice, we can provide the guarantees stated in the lemma.

	Note that we can combine the update step of the incumbent solution and the pool of a given run with the next run's distribution of elements into a single round of communication. Specifically, each machine computes a new random assignment for each element of its sample $X_{i, r}$, assigns all of its new pool elements to \emph{all} machines, and sends its solution to the designated machine. Thus each run corresponds to a round of communication. In each round, a machine communicates its sample $X_{i, r}$, which has size $O(n/m) = O(\sqrt{n s}/(\eps \alpha))$ with high probability, and the sets $\AlgSol(X_{i, r} \cup C_{r - 1})$ and $\AlgRel(X_{i, r} \cup C_{r - 1})$ that have size $O(s)$ to all other machines. Thus the total amount that a machine communicates is $O(\sqrt{n s}/(\eps \alpha) + s\cdot gm) = O(\sqrt{n s}/(\eps \alpha))$ with high probability, and the total amount that all machines communicate is $O(n + n/m \cdot gm)  = O(n/(\eps\alpha))$.
	
	In every round, the space used on a given machine is the size of its sample $X_{i, r}$, which is $O(n/m) = O(\sqrt{n s}/(\eps \alpha))$ with high probability; the size of the incumbent solution, which is $O(s)$; and the size of the pool, which is $O(gm \cdot s / \eps) = O(\sqrt{n s} / \eps)$. Therefore the total amount of space used on each machine is $O(\sqrt{n s} / (\eps \alpha))$ with high probability.
\end{proof}

\smallskip
Thus it remains to analyze the quality of the solution constructed by the algorithm. In the remainder of this section, we show that, if $\Alg$ satisfies the $\alpha$-approximation and consistency properties defined above, the parallel algorithm $\ParallelAlg$ achieves an $(\alpha - O(\eps))$ approximation. For simplicity, in this section we assume that $\Alg$ is deterministic; in Section~\ref{sec:parallel-framework-rand}, we extend our approach to the setting in which \Alg is randomized. 
We start by introducing some notation. Let $\mathcal{V}(1/m)$ denote the distribution over random subsets of $V$ where each element is included independently with probability $1/m$. Let $\opt$ be an optimal solution. Recall that $X_{i, r} \sim \mathcal{V}(1/m)$ is the random sample placed on machine $i$ at the beginning of run $r$ and $C_{r - 1}$ is the pool of elements at the beginning of run $r$. The following theorem is the crux of our analysis.

\begin{theorem} \label{thm:parallel-main}
  Consider a run $r \ge 1$ of the algorithm. Let $\hC_{r - 1} \subseteq V$. Then one of the following must hold:
	\begin{enumerate}[$(1)$]
    \item $\Ex_{X_{1,r}}[f(\AlgSol(C_{r - 1} \cup X_{1, r})) \mid C_{r - 1} = \hC_{r - 1}] \geq (1 -~\eps)^2 \alpha \cdot f(\opt)$, or
    \item $\Ex[f(C_r \cap \opt) \mid C_{r - 1} = \hC_{r - 1}] - f(\hC_{r-1} \cap\opt) \geq {\eps \over 2} \cdot f(\opt)$.
	\end{enumerate}
\end{theorem}

Intuitively, Theorem~\ref{thm:parallel-main} shows that, in expectation, if we have not found a good solution on some machine after $O(1/\epsilon)$ runs, then the current pool $C$, available to every machine, must satisfy $f(C \cap \opt) = f(\opt)$, and so each machine in the next run will in fact return a solution of quality at least $\alpha f(\opt)$. The following theorem, whose proof we give in Section~\ref{sec:omitted}, makes this formal.

\begin{theorem}\label{thm:parallel-approx}
  \ParallelAlg achieves an $(1 - \eps)^3 \alpha$ approximation with constant probability.
\end{theorem}

We devote the rest of this section to the proof of Theorem~\ref{thm:parallel-main}.  Consider a run $r$ of the algorithm. Let $\hC_{r - 1} \subseteq V$. In the following, we condition on the event that $C_{r - 1} = \hC_{r - 1}$.

For each element $e \in V$, let $p_r(e) = \Pr_{X \sim \mathcal{V}(1/m)}[e \in \AlgRel(\hC_{r - 1} \cup X \cup \set{e})]$ if $e \in \opt \setminus \hC_{r - 1}$, and $0$ otherwise.
As shown in the following lemma, the probability $p_r(e)$ gives us a handle on the probability that $e$ is in the union of the relevant sets.

\begin{lemma} \label{lem:probabilities}
  For each element $e \in \opt \setminus \hC_{r - 1}$, \[\Pr[e \in \cup_{1 \leq i \leq gm} \AlgRel(\hC_{r - 1} \cup X_{i, r})] = 1 - (1 - p_r(e))^g,\] where $g$ is the number of groups into which the machines are partitioned.
\end{lemma}
\begin{proof}
  For each group $G_j$, we can show that $e$ is \emph{not} in the union of the relevant sets for that group with probability $1 - p_r(e)$. Since different groups have independent partitions, $e$ is not in the union of the relevant sets for all machines with probability $(1 - p_r(e))^g$, and the lemma follows.  More precisely, for each group $G_j$, let $\overline{Y}_j$ be the event that $e \notin \bigcup_{i \in G_j} \AlgRel(\hC_{r - 1} \cup X_{i, r})$.  Let $G_{j,\ell}$ denote the $\ell$th machine in $G_j$.  We have
  \begin{multline*}
    \Pr\left[\overline{Y}_j \right]
    = {1 \over m} \sum_{\ell = 1}^m \Pr[\overline{Y}_j \mid \text{$e$ is on $G_{j,\ell}$}]
    = {1 \over m} \sum_{\ell = 1}^m \Pr_{X_{\ell, r}} [e \notin \AlgRel(\hC_{r - 1} \cup X_{\ell, r}) \mid e \in X_{\ell, r}]\\
    = {1 \over m} \sum_{\ell = 1}^m \Pr_{X \sim \mathcal{V}(1/m)} [e \notin \AlgRel(\hC_{r - 1} \cup X \cup \set{e})]
    = 1 - p_r(e),
  \end{multline*}
where the first equality follows from the fact that $e$ assigned to a machine in $G_j$ chosen independently and uniformly at random, and the third from the fact that the distribution of $X_{\ell, r} \sim \mathcal{V}(1/m)$ conditioned on $e \in X_{\ell, r}$ is identical to the distribution of $X \cup \set{e}$ with $X \sim \mathcal{V}(1/m)$.  Since the events $\set{Y_j \colon 1 \leq j \leq g}$ are mutually independent, $\Pr[ \underset{1 \leq j \leq N}{\wedge} \overline{Y}_j ] = \prod_{j = 1}^g   \Pr[\overline{Y}_j ] = (1 - p_r(e))^g$.
\end{proof}

Returning to the proof of Theorem \ref{thm:parallel-main}, we define a partition $(P_r, Q_r)$ of $\opt \setminus \hC_{r - 1}$ as follows:
\begin{align*}
  P_r &= \{e \in \opt \setminus \hC_{r - 1} \colon p_r(e) < \eps\} &
  Q_r &= \{e \in \opt \setminus \hC_{r - 1} \colon p_r(e) \geq \eps\}
\end{align*}
The following subsets of $P_r$ and $Q_r$ are key to our analysis (recall that $X_{i, r}$ is the random sample placed on machine $i$ at the beginning of the run $r$):
\begin{align*}
  P'_r &= \big\{e \in P_r \colon e \notin \AlgRel(\hC_{r - 1} \cup X_{1, r} \cup \set{e})\big\} &
  Q'_r &= Q_r \cap \big(\cup_{i = 1}^{gm} \AlgRel(\hC_{r - 1} \cup X_{i, r}) \big).
\end{align*}
Note that each element $e \in P_r$ is in $P'_r$ with probability $1 - p_r(e) \geq 1 - \eps$.  Further, by Lemma~\ref{lem:probabilities}, each element $e \in Q_r$ is in $Q'_r$ with probability $1 - (1 - p_r(e))^g \geq 1 - {1 \over e} \geq {1 \over 2}$.

It follows from the definition of $P'_r$ and the consistency property of $\Alg$ that
  \[ \AlgSol(\hC_{r - 1} \cup X_{1, r}) = \AlgSol(\hC_{r - 1} \cup X_{1, r} \cup P'_r).\]
Let $\opt_{r-1} = \hC_{r-1} \cap \opt$ be the part of $\opt$ in this iteration's pool.  Then, since \Alg is an $\alpha$ approximation and $P'_r \cup \opt_{r-1} \subseteq \opt$ is a feasible solution, we have
   \[ f(\AlgSol(\hC_{r - 1} \cup X_{1, r})) \geq \alpha \cdot f(P'_r \cup \opt_{r-1}). \]
Taking expectations on both sides, we have:
\begin{equation}
\Ex[f(\AlgSol(\hC_{r - 1} \cup X_{1, r}))] \geq \alpha \cdot \Ex[f(P'_r \cup \opt_{r-1})] \geq (1 - \eps) \alpha \cdot f(P_r \cup \opt_{r-1}), \label{eq:greedy}
\end{equation}
where the final inequality follows from Lemma \ref{lem:lovasz-1}, since $f$ is monotone when restricted to $\opt \supseteq P_r \cup \opt_{r-1}$, and $P'_r$ contains every element of $P_r$ with probability at least $(1 - \epsilon)$.

Note that $Q_r' \subseteq (\opt \cap C_r) \setminus \opt_{r-1}$.  As before, $f$ is monotone when restricted to $\opt$. Additionally, $Q'_r$ contains every element of $Q_r$ with probability at least $1/2$.  Thus,
\begin{equation*}
  \Ex[f(C_r \cap \opt) \mid C_{r - 1} = \hC_{r - 1}] \geq \Ex[f(Q'_r \cup \opt_{r-1})] \geq {1 \over 2} \cdot f(Q_r\cup \opt_{r-1}) + \frac{1}{2} \cdot f(\opt_{r-1}), 
\end{equation*}
where the final inequality follows from Lemma \ref{lem:lovasz-1}.  Rearranging this inequality using the condition $C_{r-1} = \hC_{r-1}$ and the definition $\opt_{r-1} = \hC_{r-1} \cap \opt$ we obtain:
\begin{multline}
  \Ex[f(C_r \cap \opt) - f(C_{r - 1} \cap \opt) \mid C_{r - 1} = \hC_{r - 1}] \geq {1 \over 2} \left(f(Q_r \cup \opt_{r-1}) - f(\opt_{r-1}) \right) \\
  \geq {1 \over 2} \left(f(P_r \cup Q_r \cup \opt_{r-1}) - f(P_r \cup \opt_{r-1}) \right)
  = {1 \over 2} \left(f(\opt) - f(P_r \cup \opt_{r-1}) \right), \label{eq:cr-increase}
\end{multline}
where the second inequality follows from submodularity.

Now, if $f(P_r \cup (\hC_{r - 1} \cap \opt)) \geq (1 - \eps) \cdot f(\opt)$ then this fact together with (\ref{eq:greedy}) imply that the first property in the statement of Theorem~\ref{thm:parallel-main} must hold. Otherwise, $f(\opt) - f(P_r \cup (\hC_{r - 1} \cap \opt)) \geq \eps \cdot f(\opt)$; this fact together with (\ref{eq:cr-increase}) implies that the second property must hold. 

\medskip
This completes the description of our generic approach. In the following sections, we instantiate the algorithm $\Alg$ with the standard Greedy algorithm and a heavily discretized Continuous Greedy algorithm, and obtain our main results stated in the introduction.

\section{A Parallel Greedy Algorithm}
\label{sec:parallel-greedy}

In this section, we combine the generic approach from Section~\ref{sec:parallel-framework} with the standard greedy algorithm, and give our results for monotone maximization stated in Theorem~\ref{thm:parallel-greedy-main}.

We let $\Alg$ be the standard $\Greedy$ algorithm. We let $\AlgRel(N) = \AlgSol(N) = \Greedy(N)$. It was shown in previous work that the \Greedy algorithm satisfies the consistency property.
\begin{lemma}[\cite{Barbosa2015}, Lemma 2] \label{lem:greedy-rejected-elements}
	Let $A \subseteq V$ and $B \subseteq V$ be two disjoint subsets
	of $V$. Suppose that, for each element $e \in B$, we have
	$\Greedy(A \cup \set{e}) = \Greedy(A)$. Then $\Greedy(A \cup B) =
	\Greedy(A)$.	
\end{lemma}
Informally, this simply means that if \Greedy rejects some element $e$ when presented with input $A \cup \{e\}$, then adding other similarly rejected elements to $A \cup \{e\}$ cannot cause $e$ to be accepted. This allows us to immediately apply the result from Section~\ref{sec:parallel-framework} and obtain the following result.

\begin{theorem}\label{thm:parallel-greedy-main}
Let $f : 2^V \to \R_+$ be a submodular function, and $\I \subseteq 2^V$ be a hereditary set system. For any $\epsilon > 0$ there is a randomized distributed  $O(1/\epsilon)$-round  algorithm that can be implemented in the model described in Section~\ref{sec:model}. The algorithm is an $(\alpha - O(\eps))$-approximation with constant probability for the problem $\max_{S \in \I}f(S)$, where $\alpha$ is the approximation ratio of  the standard, sequential greedy algorithm for the same problem.
\end{theorem}

\section{A Parallel Continuous Greedy Algorithm}
\label{sec:parallel-dcgreedy}

\begin{figure}[!t]
\begin{center}
\begin{minipage}{0.38\linewidth}
\begin{algorithm}[H]
\KwIn{$N \subseteq V$}
$\bx(0) \gets \zero$\;
\For{$t \gets 1$ to $1/\eps$}{
  $\by(t) \gets \GreedyStep(N, \bx(t))$\;
  $\bx(t) \gets \bx(t-1) + \by(t)$\;
}
$S \gets \mathsf{SwapRounding}(\bx(1/\eps), \I)$\;
Let $T$ be the support of $\bx(1/\eps)$ \;
\Return $(S, T)$ \;
\caption{Discretized Continuous Greedy (\DCGreedy).}
\label{alg:dcgreedy}
\end{algorithm}
\end{minipage}
\hspace{0.1in}
\begin{minipage}{0.57\linewidth}
\begin{algorithm}[H]
\KwIn{$N \subseteq V$, $\bx \in [0, 1]^N$}
$W \gets \emptyset$, $\by \gets \zero$\;
\InfLoop{}{
  $D \gets \{e \in N \setminus W : W \cup \{e\} \in \I$\}\;
  \ForEach{$e \in D$}
  {
    $w_e \gets \Ex[f(R(\bx + \by) \cup \set{e}) - f(R(\bx + \by))]$
  }
  Let $e^* = \arg\max_{e \in D} w_e$\;
  \lIf{$D = \emptyset$ or $w_{e^*} < 0$}{\Return $\by$}
  \Else{$y_{e^*} \gets y_{e^*} + \eps(1 - x_{e^*})$\;
    $W \gets W \cup \{e^*\}$\;
  }
}
\caption{Greedy Update Step (\GreedyStep).}
\label{alg:standard-greedy}
\end{algorithm}
\end{minipage}
\end{center}
\caption{The discretized continuous greedy algorithm.  On line 5 of Algorithm~\ref{alg:standard-greedy}, for a vector $\bz \in [0, 1]^N$, we use $R(\bz)$ to denote a random subset of $N$ that contains each element $e$ independently with probability $z_e$. The weights on line 5 cannot be computed exactly in polynomial time, but they can be efficiently approximated using random samples.}
\end{figure}

For monotone maximization subject to a matroid constraint, Theorem~\ref{thm:parallel-greedy-main} guarantees only a $(1/2 - \eps)$ approximation, due to the limitations of the standard greedy algorithm. We obtain a nearly optimal $(1 - 1/e - \eps)$ approximation by instantiating the framework in Section~\ref{sec:parallel-framework} with the \DCGreedy algorithm shown in Algorithm~\ref{alg:dcgreedy}.

The \DCGreedy algorithm is a heavily discretized version of the measured continuous greedy approach of \cite{Feldman2011a}, and it first constructs an approximate fractional solution to the problem $\max_{\bx \in P(\I)} F(\bx)$ of maximizing the multilinear extension $F$ of $f$ subject to the constraint that $\bx$ is in the matroid polytope $P(\I)$, and then rounds the fractional solution without loss using pipage rounding or swap rounding \cite{Ageev2004,Chekuri2010}.  

In this section, we combine the generic approach from Section~\ref{sec:parallel-framework} with the \DCGreedy algorithm. We use \DCGreedy as \Alg; the relevant set $\AlgRel(N)$ is the set of elements in the support of the fractional solution $\bx(1/\eps)$, and $\AlgSol(N)$ is the integral solution obtained by rounding $\bx(1/\eps)$.

Note that it is necessary to ensure that the fractional solution has small support so that the size of $\AlgRel(N)$ is small. We achieve this by heavily discretizing the continuous greedy algorithm, thereby limiting the number of support updates performed in lines 3 and 4 of \DCGreedy.  Unfortunately, performing this discretization naively introduces an error in the approximation that is too large.  Thus, we make use of a key idea from \cite{Badanidiyuru2014}, which can be applied in the case of a matroid constraint. This allows us to show the following lemma whose proof is deferred to Section~\ref{sec:dcgreedy-analysis}.

\begin{lemma}
\label{lem:dcgreedy-analysis}
  The \DCGreedy algorithm achieves an $(1 - 1/e - O(\eps))$ approximation for monotone functions and an $(1/e - O(\eps))$ approximation for non-monotone functions. 
\end{lemma}

The lemma above provides us with the desired approximation guarantees for \DCGreedy, and thus it remains to show the consistency property. Before doing so, we must address how the weights are computed on line 5 of the \GreedyStep algorithm (see Algorithm~\ref{alg:standard-greedy}). Computing the weights exactly requires exponential time, but they can be approximated in polynomial time using random samples. In order to illustrate the main ideas behind the proof of consistency, \emph{we assume that the weights are computed exactly}, since this will keep the algorithm deterministic. In the Appendix, we remove this assumption and we analyze the resulting randomized algorithm using an extension of our framework.

\begin{lemma} \label{lem:nonmonotone-rejected-elements}
	Let $A$ and $B$ be two disjoint subsets
	of $V$. Suppose that, for each element $e \in B$, we have
	$\DCGreedyRel(A \cup \set{e}) = \DCGreedyRel(A)$. Then $\DCGreedySol(A \cup B) =
	\DCGreedySol(A)$.	
\end{lemma}
\begin{proof}
  We will show that the \GreedyStep algorithm picks the same set $W$ on input $(A, \bx)$ and $(A \cup B, \bx)$, which implies the lemma. Suppose for contradiction that the algorithm makes different choices on input $(A, \bx)$ and $(A \cup B, \bx)$. Consider the first iteration where the two runs differ, and let $e$ be the element added to $W$ in that iteration on input $(A \cup B, \bx)$. Note that $e \notin A$ and thus we have $e \in B$. But then $e$ will be added to $W$ on input $(A \cup \set{e}, \bx)$. Thus $e \in \DCGreedyRel(A \cup \set{e})$, which contradicts the fact that $e \in B$.
\end{proof}

Thus we can apply the result from Section~\ref{sec:parallel-framework} and obtain the following result.

\begin{theorem} \label{thm:parallel-cg-main}
Let $f : 2^V \to \R_+$ be a submodular function, and $\I \subseteq 2^V$ be a matroid.  For any $\epsilon > 0$ there is a randomized distributed  $O(1/\epsilon)$-round  algorithm that can be implemented in the model described in Section~\ref{sec:model}.  The algorithm is an $(\alpha - O(\eps))$-approximation with constant probability for the problem $\max_{S \in \I}f(S)$, where $\alpha$ is $(1 - 1/e)$ for monotone $f$ and $1/e$ for general $f$.
\end{theorem}

\section{Faster Algorithms}
\label{sec:two-round}
In this section, we build on the techniques from the previous sections to give a distributed algorithm for non-monotone maximization that requires only two rounds, rather than $O(1/\eps)$ rounds, and achieves an improved approximation guarantee over the two-round algorithm proposed in \cite{Barbosa2015}. In the case of non-monotone maximization over a matroid, we show that our techniques can be used to obtain a new, fast \emph{sequential} algorithm as well.

\subsection{Two-Round Algorithms For Non-Monotone Maximization}
\label{sec:two-round-algorithms}

We first give an improved two-round algorithm for non-monotone maximization subject to a any hereditary constraint.  The algorithm is similar to that of \cite{Barbosa2015} for \emph{monotone} maximization; perhaps surprisingly, we show that this approach achieves a good approximation even for non-monotone functions. We randomly partition the elements onto the $m$ machines, and run $\Greedy$ on the elements $V_i$ on machine $i$ to pick a set $S_i$. We place the sets $S_i$ on a single machine and we run any algorithm \Alg on $B \coloneqq \bigcup_i S_i$ to find a solution $T$. We return the best solution amongst $S_1, \dots, S_m, T$.

We analyze the algorithm for any hereditary constraint $\I$ for which the \Greedy algorithm satisfies the following property (for some $\gamma$), which we refer to as the strong greedy property:
\begin{equation}
	\forall S \in \I: f(\Greedy(V)) \ge \gamma \cdot f(\Greedy(V) \cup S)
	\tag{\ensuremath{\mathrm{GP}}}
	\label{eq:greedy-strong-property}
\end{equation}

By the standard \Greedy analysis, we have $\gamma = 1/2$ for a matroid constraint and $\gamma = 1/(p + 1)$ for a $p$-system constraint.

\begin{theorem}
\label{thm:two-round}
Suppose that \Greedy satisfies the strong greedy property with constant $\gamma$ and let \Alg be any $\beta$-approximation for the problem $\max_{S \in \I}f(S)$. Then there is a randomized, two-round distributed algorithm that achieves a $( 1 - {1 \over m}){\beta \gamma  \over \beta + \gamma}$ approximation in expectation for $\max_{S \in \I}f(S)$.
\end{theorem}
\begin{proof}
For each element $e$, we let probability $p_e = \Pr_{X \sim \mathcal{V}(1/m)}[e \in \Greedy(X \cup \set{e})]$, if $e \in \opt$,
and $0$ otherwise.  Then, let $\bp \in [0, 1]^V$ denote the vector whose entries are given by the probabilities $p_e$.

We first analyze the expected value of the \Greedy solution $S_1$. Let 
\[O = \set{e \in \opt \colon e \notin \Greedy(V_1 \cup \{e \})}.\]
By Lemma~\ref{lem:greedy-rejected-elements}, $\Greedy(V_1 \cup O) = \Greedy(V_1) = S_1$, and by \eqref{eq:greedy-strong-property}, $f(S_1) \geq \gamma \cdot f(S_1 \cup O)$.
Therefore
\begin{align}
	\Ex[f(S_1)] &\geq \gamma \cdot \Ex[f(S_1 \cup O)] \nonumber\\
	&= \gamma \cdot \Ex[f^-(\one_{S_1 \cup O})] \nonumber\\
	&\geq \gamma \cdot f^{-}(\Ex[\one_{S_1 \cup O}]) \nonumber\\
	&= \gamma \cdot f^-(\Ex[\one_{S_1}] + (\one_{\opt} - \bp)).
\label{eq:first-machine}
\end{align}
On line three, we have used the fact that $f^-$ is convex and on line four we have used the fact that $\Ex[\one_{S_1 \cup O}] = \Ex[\one_{S_1}] + (\one_{\opt} - \bp)$.

Now consider the solution $T$.  Since \Alg is a $\beta$-approximation, we have
\begin{align}
	\Ex[f(T)] &\geq \beta \cdot \Ex[f(B \cap \opt)] \nonumber\\
	&= \beta \cdot \Ex[f^-(\one_{B \cap \opt})] \nonumber\\
	&\geq \beta \cdot f^-(\Ex[\one_{B \cap \opt}]) \nonumber\\
	&= \beta \cdot f^-(\bp). \label{eq:last-machine}
\end{align}
Similarly to above, we have used the convexity of $f^-$ and the fact that $\Ex[\one_{B \cap \opt}] = \bp$.  

By combining (\ref{eq:first-machine}) and (\ref{eq:last-machine}), and using convexity of $f^-$, we obtain
\begin{equation*}
	{1 \over \gamma} \Ex[f(S_1)] + {1 \over \beta} \Ex[f(T)] \geq f^-(\Ex[\one_{S_1}] + (\one_{\opt} - \bp)) + f^-(\bp)
	\geq 2 \cdot f^-\left({\Ex[\one_{S_1}] + \one_{\opt} \over 2}\right).
\end{equation*}
Since $S_1 \subseteq V_1$ and $V_1$ is a $1/m$ sample of $V$, we have $\Ex[\one_{S_1}] \leq {1 \over m} \cdot \one_V$. Therefore, using the definition of $f^-$ and the non-negativity of $f$, we obtain
\[ 2 \cdot f^-\left({\Ex[\one_{S_1}] + \one_{\opt} \over 2} \right) \geq \left(1 - {1 \over m} \right) f(\opt). \]

Thus
\[ \max\{\Ex[f(S_1)], \Ex[f(T)]\} \geq \left(1 - {1 \over m} \right) {\beta \gamma \over \beta + \gamma} \cdot f(\opt). \]
\end{proof}

\mypar{Examples of results.}
We conclude this section with some examples of approximation guarantees that we can obtain using Theorem~\ref{thm:two-round}. For a matroid constraint, we have $\gamma = 1/2$ and, if we use the measured Continuous Greedy algorithm for \Alg, we have $\beta = 1/e$; thus we obtain a $\left(1 - {1 \over m} \right) {1 \over 2 + e}$ approximation. We remark that, for a cardinality constraint, one can strengthen the proof of Theorem~\ref{thm:two-round} slightly and obtain a $\left(1 - {1 \over m} \right) {1 \over e} \left(1 - {1 \over e} \right)$ approximation; we give the details in Section~\ref{sec:two-round-cardinality}.

For a $p$-system constraint, we have $\gamma = {1/ (p + 1)}$. We can use the algorithm of Gupta \etal \cite{Gupta2010} for $\Alg$ that achieves an approximation $\beta = 3 / \left(2p + 4 + {2 \over p}\right)$ when combined with the algorithm of \cite{Buchbinder2012} for unconstrained non-monotone maximization. Thus we obtain a $3\left(1 - {1 \over m}\right)/ \left(5p + 7 + {2 \over p}\right)$ approximation.

\begin{figure}[!t]
\begin{center}
\begin{minipage}{0.5\linewidth}
\begin{algorithm}[H]
\label{alg:dtreshgreedy}
\caption{Descending Thresholds Greedy algorithm (\DTreshGreedy).}
\KwIn{$N\subseteq V$}
$S \gets \emptyset$, $d \gets \max_{e \in N} f(\set{e})$ \;
\For{$w = d$; $w \geq {\eps \over n} d$; $w \gets w (1 - \eps)$}{
  \ForEach{$e \in N$} {
    \If{$S \cup \{e\} \in \I$ \; \hspace{0.4em}\lAnd $f(S \cup \{e\}) - f(S) \geq w$}{
      $S \gets S \cup \set{e}$
    }
  }
}
\textbf{return} $S$ \;
\end{algorithm}
\end{minipage}
\caption{The Descending Thresholds Greedy algorithm of \cite{Kumar2013,Badanidiyuru2014}.}
\label{fig:dtgreedy}
\end{center}
\end{figure}

\subsection{A Fast Sequential Algorithm for Matroid Constraints}
\label{sec:fast-sequ-algor}

We now show how our approach can be used to obtain a fast \emph{sequential} algorithm for non-monotone maximization subject to a matroid constraint. The analysis given in Theorem~\ref{thm:two-round} only relies on the following two properties of the \Greedy algorithm: it satisfies (\ref{eq:greedy-strong-property}) and Lemma~\ref{lem:greedy-rejected-elements}. Thus we can replace the \Greedy algorithm by any algorithm satisfying these two properties. In particular, the Descending Thresholds Greedy (shown in  Figure~\ref{fig:dtgreedy} as \DTreshGreedy) of \cite{Kumar2013,Badanidiyuru2014} satisfies these conditions with $\gamma = 1/2 - \eps$.  

Our algorithm proceeds as follows.  We randomly partition the elements into $m \coloneqq 1/\eps$ samples $V_1, V_2, \dots, V_m$. On each sample, we run the Descending Thresholds Greedy algorithm on $V_i$ to obtain a solution $S_i$. Let $A \coloneqq \argmax_{i \in [m]} f(S_i)$ and $B \coloneqq \bigcup_i S_i$. Then, $|B| \le k/\eps$, where $k$ is the rank of the matroid.  We run any $\beta$-approximation algorithm $\Alg$ on $B$ to find a solution $B'$, and we return the better of $A$ and $B'$.  We obtain the following result.

\begin{theorem}\label{thm:fast-matroid}
There is a sequential, randomized $(\frac{1}{2+e} - \eps)$-approximation algorithm for the problem $\max_{S \in \I}f(S)$, where $\I$ is any matroid constraint, running in time $O({n \over \eps} \log{n}) + \poly({k \over \eps})$.
\end{theorem}
\begin{proof}
The running time of the Descending Thresholds Greedy algorithm on a ground set of size $s$ is $O({s \over \eps} \log({s \over \eps}))$. Each random sample has size $O(\eps n)$ with high probability, and thus the total time needed to construct $B$ is $O({n \over \eps} \log{n})$ with high probability. It follows from the analysis in Theorem~\ref{thm:two-round} that the best of the two solutions $A$ and a $\beta$-approximation to $\max_{S \subseteq B : S \in \I} f(S)$ is a ${1 \over 2 + {1 \over \beta}} - \eps$ approximation.  We can then use any $1/e$-approximation algorithm as \Alg.
\end{proof}

\subsection{Monotone Maximization with a Cardinality Constraint}
\label{sec:monot-maxim-with}

Here, we show how the previous techniques give a simple two-round algorithm that achieves a $1/2 - \eps$ approximation for monotone maximization subject to a cardinality constraint.

\begin{theorem} \label{thm:half-cardinality}
There is a randomized, two-round, distributed algorithm achieving a ${1 \over 2} - \eps$ approximation in expectation for $\max_{S : |S| \le k}f(S)$, where $f$ is a monotone function.
\end{theorem}

\mypar{The algorithm.}
Let $\eps > 0$ be a parameter. The algorithm uses $\Theta(\log(1/\eps)/\eps)$ groups of machines with $m$ machines in each group (and thus the total number of machines is $O(m \log(1/\eps)/\eps)$). 

We randomly distribute the ground set $V$ to the machines as follows.
Amongst each group of $m$ machines, we partition $V$ uniformly at random; each element $e$ chooses an index $i \in [m]$ uniformly and independently at random and is assigned to the $i$th machine in the group.  We do this separately for each group of machines, i.e., each element appears on exactly one machine in each group.

We run \Greedy on each of the machines to select a set of $k$ elements. Let $S$ be the union of all of the \Greedy solutions. We place $S$ on a single machine together with a random sample $X \sim \script{V}(1/m)$. On this machine, we pick the final solution as follows. For each value $a \in \set{0, 1, \dots, k}$, we select a solution $T_a$ as follows.  Let $\Greedy(f,X,a)$ denote the first $a$ elements chosen from the random sample $X$ using the greedy algorithm on objective function $f$.

Then, let $T^1_a = \Greedy(f, X, a)$ and define $g(A) = f(T_a^1 \cup A) - f(T_a^1)$ for each $A \subseteq V$.  Note that $g$ is a non-negative, monotone submodular function. Let $T^2_a = \Greedy(g, S, k - a)$; that is, we pick $k - a$ elements from $S$ using the Greedy algorithm with the function $g$ as input. We set $T_a = T^1_a \cup T^2_a$. The final solution $T$ is the better of the $k + 1$ solutions $T_a$, where $a \in \set{0, 1, \dots, k}$.

\mypar{The analysis.}
In the following, we show that the algorithm above is a $1/2 - \eps$ approximation. For each element $e$, we define a probability $p_e=  \Pr_{X \sim \mathcal{V}(1/m)}[e \in \Greedy(X \cup \set{e})]$, if $e \in \opt$ and 0 otherwise.  We define a partition $(O_1, O_2)$ of $\opt$ as follows:
\begin{align*}
	O_1 &= \set{e \in \opt \mid p_e < \eps},&
	O_2 &= \set{e \in \opt \mid p_e \geq \eps}.
\end{align*}
Let $a = \card{O_1}$ and let
	\[ 	O'_1 = \set{e \in O_1 \mid e \notin \Greedy(f, X \cup \set{e}, a)}.\]
By the consistency  property of the greedy algorithm (Lemma~\ref{lem:greedy-rejected-elements}),
  \[ T_a^1 = \Greedy(f, X, a) = \Greedy(f, X \cup O'_1, a). \]
Additionally, for a cardinality constraint, \Greedy satisfies (\ref{eq:greedy-strong-property}) with $\gamma = 1/2$ (see Subsection~\ref{sec:two-round-algorithms}). Therefore
\begin{align}
	f(T_a^1) &\geq {1 \over 2} f(T_a^1 \cup O'_1), \label{eq1}\\
	g(T_a^2) &\geq {1 \over 2} g(T_a^2 \cup (O_2 \cap S)). \label{eq3c}
\end{align}
The inequality (\ref{eq3c}) can be rewritten as
\begin{equation}
	f(T_a^1 \cup T_a^2) - f(T_a^1) \geq {1 \over 2} (f(T_a^1 \cup T_a^2 \cup (O_2 \cap S)) - f(T_a^1)) \label{eq2}.
\end{equation}
Adding (\ref{eq1}) and (\ref{eq2}), we obtain
\begin{align*}
	f(T_a) &\geq {1 \over 2} (f(T_a^1 \cup O'_1) + f(T_a^1 \cup T_a^2 \cup (O_2 \cap S)) - f(T_a^1)) \notag\\
	&\geq {1 \over 2} f(T_a \cup O'_1 \cup (O_2 \cap S)) \notag\\
	&\geq {1 \over 2} f(O'_1 \cup (O_2 \cap S)) \label{eq3d},
\end{align*}
where the last two inequalities follow from submodularity and monotonicity.
Note that each element $e \in O_1$ is in $O'_1$ with probability $1 - p_e \geq 1 - \eps$. Each element $e \in O_2$ is in the union of the \Greedy solutions from a given group of machines with probability $p_e \geq \epsilon$; since there are $\Theta(\log(1/\epsilon)/\epsilon)$ groups of machines and the groups have independent partitions, $e$ is in $S$ with probability at least $1 - \epsilon$. Therefore
\begin{equation*}
	\Ex[\one_{O'_1 \cup (O_2 \cap S)}] \geq (1 - \epsilon) \one_{\opt}. \label{eq4}
\end{equation*}
Thus
  \[ \Ex[f(T_a)] \geq {1 \over 2} f^-(\Ex[\one_{O'_1 \cup (O_2 \cap S)}]) \geq (1 - \epsilon) {1 \over 2} f(\opt).\]
In the last inequality, we have used that if $x \geq y$ component-wise and $f$ is monotone, $f^-(x) \geq f^-(y)$.

\bigskip\noindent
{\bf Acknowledgments.} This work was done in part while A.E. was with the Computer Science department at the University of Warwick and a visitor to the Toyota Technological Institute at Chicago, H.N. was with the Toyota Technological Institute at Chicago, and J.W. was with the Computer Science department at the University of Warwick. J.W. was supported by EPSRC grant EP/J021814/1 and ERC Starting Grant 335288-OptApprox.
 
\bibliographystyle{plain}
\bibliography{references}

\newpage
\appendix

\section{Proof of Theorem~\ref{thm:parallel-approx}}
\label{sec:omitted}

\begin{appendix-theorem}{thm:parallel-approx}
  \ParallelAlg achieves an $(1 - \eps)^3 \alpha$ approximation with constant probability.
\end{appendix-theorem}
\begin{proof}
Let $R = c/\eps$ be the total number of runs, and $\script{C} = (C_0, C_1, \dots, C_R)$.  Let $I_r(C_{r - 1}) \in \set{0, 1}$ be equal to $1$ if and only if
\[ \Ex_{X_{1,r}}[f(\AlgSol(C_{r - 1} \cup X_{1, r}))] \geq (1 - \eps)^2 \alpha \cdot f(\opt). \]
  Let
  \begin{align*}
    \Phi_r(\script{C}) &= I_r(C_{r - 1}) + {2 (f(C_r \cap \opt) - f(C_{r - 1} \cap \opt)) \over \eps f(\opt)},\\
    \Phi(\script{C}) &= \sum_{r = 1}^{R} \Phi_r(C_{r - 1})
	 \leq \sum_{r = 1}^R I_r(C_{r - 1}) + {2 f(C_R \cap \opt) \over \eps f(\opt)}
	 \leq \sum_{r = 1}^R I_r(C_{r - 1}) + {2 \over \eps}.
  \end{align*}
	Taking expectation over the random choices of $\script{C}$, we have
	  $$ \Ex_{\script{C}}[\Phi(\script{C})] \le \sum_{r=1}^R \Ex[I_r(C_{r-1})] + \frac{2}{\eps}$$
	On the other hand, by Theorem~\ref{thm:parallel-main}, $\Ex[\Phi_r(C_{r - 1})] \geq 1$ and therefore $\Ex[\Phi(\script{C})] \geq R$. Thus
  \[ {2 \over \eps} + \sum_{r = 1}^R \Ex[I_r(C_{r - 1})] \geq \Phi(\script{C}) \geq R.\]
  Since $R > 6/\eps$, we have
	\[ \sum_{r = 1}^R \Ex[I_r(\hC_{r - 1})] \geq {2R \over 3}.\]
	Therefore, with probability at least $2/3$, there exists a run $r$ such that $I_r(C_{r-1})=1$. Fix the randomness up to the first such run, i.e., condition on a fixed $C_{r-1} = \hC_{r-1}$ such that $I_r(\hC_{r-1}) = 1$ and $C_r, \ldots, C_R$ remain random. Assume for contradiction that with probability at least $1-\eps \alpha (1-\eps)^2$ over the choices of $X_{1,r}$, 
	$$f(\AlgSol(C_{r - 1} \cup X_{1, r})) < (1 - \eps)^3 \alpha \cdot f(\opt).$$
  Then we have
  \begin{align*}
    \Ex[f(\AlgSol(C_{r - 1} \cup X_{1, r}))]  &< (\eps \alpha (1-\eps)^2+ (1-\eps \alpha(1-\eps)^2)(1-\eps)^3\alpha)f(\opt)\\
    &=\left(\eps + (1-\eps\alpha(1-\eps)^2)(1-\eps)\right)(1-\eps)^2 \alpha f(\opt)\\
    &<  (1-\eps)^2 \alpha f(\opt),
 \end{align*}
contradicting our assumption on $C_{r-1}$. Thus, with probability at least $\eps \alpha (1-\eps)^2$, we have \[f(\AlgSol(C_{r - 1} \cup X_{1, r})) \ge (1 - \eps)^3 \alpha \cdot f(\opt).\] Notice that the above argument applies not only to machine $1$ in run $r$ but also the first machine in each of the $g$ groups in the same run $r$ and their random samples $X_{i, r}$ are independent. Thus, since $g \geq c /(\eps\alpha)$ for a sufficiently large constant $c$, with probability at least $5/6$, we have $\max_i f(\AlgSol(C_{r - 1} \cup X_{i, r})) \ge (1 - \eps)^3 \alpha \cdot f(\opt)$. Overall, the algorithm succeeds with probability at least $2/3 \cdot 5/6 = 5 / 9$.
\end{proof}

\section{A Framework for Parallelizing Randomized Algorithms}
\label{sec:parallel-framework-rand}

In this section, we extend the framework from Section~\ref{sec:parallel-framework} to the setting in which the sequential algorithm \Alg is randomized.

We represent the randomness of \Alg as a vector $\bb \sim \D$ drawn from some distribution $\D$. It is convenient to have the randomness $\bb$ given as input to the algorithm. More precisely, we assume that \Alg takes as input a subset $N \subseteq V$ and a random vector $\bb \sim \D$ and returns a pair of sets, $\AlgSol(N, \bb)$ and $\AlgRel(N, \bb)$. We assume that the size of $\bb$ depends only on the size of $V$, and hence is independent of the size of $N$. Finally, we assume that \Alg has the following properties.

\begin{enumerate}
\item ($(\alpha, \eps, \delta)$-Approximation) Let $\opt = \argmax_{S \in \I, S \subseteq V} f(S)$ be an optimal solution over the entire ground set $V$. Let $A \subseteq V$ and $\bb \sim \D$. Let $B \subseteq \opt$ be a subset such that, for each $e \in B$, $e \notin \AlgRel(A \cup \set{e}, \bb)$. We have
 \[ \Pr_{\bb \sim \D}\Big[f(\AlgSol(A \cup B, \bb)) \geq \alpha \cdot f((A \cup B) \cap \opt) - \eps f(\opt) \Big] \geq 1 - \delta.\]
 \item (Consistency) Let $\bb$ be any fixed vector. Let $A$ and $B$ be two disjoint subsets of $V$.  Suppose that, for each element $e \in B$, we have $\AlgRel(A \cup \set{e}, \bb) = \AlgRel(A, \bb)$. Then $\AlgSol(A \cup B, \bb) = \AlgSol(A, \bb)$.
\end{enumerate}
Note that our assumption that the length of $\bb$ is independent of the size of the input subset allows expressions such as $\Alg(A \cup \set{e}, \bb)$ and $\Alg(A, \bb)$ to both make sense despite the fact that $\card{A\cup \set{e}} \neq \card{A}$.

Our algorithm works exactly as that described in Section \ref{sec:parallel-framework}, with the exception that 
each machine $i$ in round $r$ now additionally samples a random vector $\bb_{i,r} \sim \D$.  Then, on each machine, we run \Alg on the set $V_{i, r} \coloneqq X_{i, r} \cup C$ of elements on the machine and obtain $\AlgSol(V_{i, r}, \bb_{i, r})$ and $\AlgRel(V_{i, r}, \bb_{i, r})$.  As in Section \ref{sec:parallel-framework}, the union $\bigcup_i \AlgRel(V_{i, r}, \bb_{i, r})$ of relevant elements is added to $C$, and the solution $S_{\mathrm{best}}$ is replaced by the best solution among $\set{\AlgSol(V_{i, r}, \bb_{i, r}) \colon 1 \leq i \leq M}$ and $S_{\mathrm{best}}$. 

In the final round we place $C$ on a single machine, sample a random vector $\bb \sim \D$, and run \Alg on $C$, and $\bb$ to obtain the solution $\AlgSol(C, \bb)$. The final solution is the best among $\AlgSol(C, \bb)$ and $S_{\mathrm{best}}$.

\mypar{Analysis.}
The number of rounds, number of machines, space per machine, and amount of communication are the same as in Section~\ref{sec:parallel-framework}.
Thus, we focus on the approximation guarantee of the parallel algorithm. Using Theorem~\ref{thm:randomized-parallel-main} instead of Theorem~\ref{thm:parallel-main}, we can then finish the analysis in almost the same way as the deterministic case. The only difference is that instead of arguing that the algorithm works well with most of the random choices for $X_{i,r}$ as before, the proof now argues that the algorithm works well with most of the random choices for both $X_{i,r}$ and $\bb_{i,r}$. Nonetheless, the same proof except for this substitution works.

\begin{theorem} \label{thm:randomized-parallel-main}
  Consider a run $r > 1$ of the algorithm. Let $\hC_{r - 1} \subseteq V$. Then one of the following must hold: 
	\begin{enumerate}[$(1)$]
    \item $\E_{X_{1,r}, \bb_{1,r}}[f(\AlgSol(C_{r - 1} \cup X_{1, r}, \bb_{1,r})) \mid C_{r - 1} = \hC_{r - 1}] \geq (\alpha-O(\eps)) \cdot f(\opt)$, or
    \item $\E[f(C_r \cap \opt) \mid C_{r - 1} = \hC_{r - 1}] - f(\hC_{r-1} \cap\opt) \geq {\eps \over 2} \cdot f(\opt)$.
	\end{enumerate}
\end{theorem}
\begin{proof}
Consider a run $r$ of the algorithm. Let $\hC_{r - 1} \subseteq V$. In the following, we condition on the event that $C_{r - 1} = \hC_{r - 1}$.

For each element $e \in V$, let $p_r(e) = \Pr_{X \sim \mathcal{V}(1/m), \bb \sim \D}[e \in \AlgRel(\hC_{r - 1} \cup X \cup \set{e}, \bb)],$ if $e \in \opt \setminus \hC_{r - 1}$, and 0 otherwise.  The proof of the following lemma is exactly the same as Lemma~\ref{lem:probabilities} and thus is omitted.

\begin{lemma} \label{lem:randomized-probabilities}
  For each element $e \in \opt \setminus \hC_{r - 1}$,
    \[\Pr[e \in \cup_{1 \leq i \leq gm} \AlgRel(\hC_{r - 1} \cup X_{i, r}, \bb_{i, r})] = 1 - (1 - p_r(e))^g,\]
  where $g$ is the number of groups into which the machines are partitioned.
\end{lemma}

We define a partition $(P_r, Q_r)$ of $\opt \setminus \hC_{r - 1}$ as follows:
\begin{align*}
  P_r &= \{e \in \opt \setminus \hC_{r - 1} \colon p_r(e) < \eps\}, & Q_r &= \{e \in \opt \setminus \hC_{r - 1} \colon p_r(e) \geq \eps\}.
\end{align*}

The following subsets of $P_r$ and $Q_r$ are key to our analysis (
recall that $X_{i, r}$ is the random sample placed on machine $i$ at the beginning of the run and $\bb_{i,r}$ is the random vector sampled by machine $i$ in round $r$):
 \begin{align*}
  P'_r &= \{e \in P_r \colon e \notin \AlgRel(\hC_{r - 1} \cup X_{1, r} \cup \set{e}, \bb_{1, r})\},
&
  Q'_r &= Q_r \cap \big(\cup_{i = 1}^{gm} \AlgRel(\hC_{r - 1} \cup X_{i, r}, \bb_{i, r})\big).
\end{align*}
Note that each element $e \in P_r$ is in $P'_r$ with probability $1 - p_r(e) \geq 1 - \eps$.  Further, by Lemma~\ref{lem:randomized-probabilities}, each element $e \in Q_r$ is in $Q'_r$ with probability $1 - (1 - p_r(e))^g \geq 1 - {1 \over e} \geq {1 \over 2}$.

It follows from the definition of $P'_r$ and the consistency property of $\Alg$ that
  \begin{equation}
\AlgSol(\hC_{r - 1} \cup X_{1, r}, \bb_{1, r}) = \AlgSol(\hC_{r - 1} \cup X_{1, r} \cup P'_r, \bb_{1, r}).
\label{eq:consistency-1}
\end{equation}
Let $\opt_{r-1} = \hC_{r-1} \cap \opt$ be the part of $\opt$ in this iteration's pool.  We apply the $(\alpha, \eps, \delta)$-approximation property with $A = \hC_{r - 1} \cup X_{1, r}$, $\bb = \bb_{1, r}$, and $B = P'_r$ to obtain
\[
  \Pr_{\bb_{1, r}}\Big[\AlgSol(\hC_{r - 1} \cup X_{1, r} \cup P'_r, \bb_{1, r}) \geq \alpha \cdot f((\hC_{r - 1} \cup X_{1, r} \cup P'_r) \cap \opt) - \eps f(\opt) \Big] \geq 1 - \delta. 
\]
Since $f$ is monotone when restricted to $\opt$, and $P_r' \cup \opt_{r-1} \subseteq (\hC_{r-1} \cup X_{1,r} \cup P_r') \cap \opt$, this inequality implies that
\[
 \Pr_{\bb_{1, r}}\Big[\AlgSol(\hC_{r - 1} \cup X_{1, r} \cup P'_r, \bb_{1, r}) \geq \alpha \cdot f(P'_r \cup \opt_{r-1}) - \eps f(\opt)\Big] \geq 1 - \delta. 
\]
Therefore, equation \eqref{eq:consistency-1} gives
\[
  \Pr_{\bb_{1, r}}\Big[\AlgSol(\hC_{r - 1} \cup X_{1, r}, \bb_{1, r}) \geq \alpha \cdot f(P'_r \cup \opt_{r-1}) - \eps f(\opt)\Big] \geq 1 - \delta. 
\]
Taking expectation on both sides gives
\begin{align}
  \E_{X_{1,r}, b_{1,r}}[f(\AlgSol(\hC_{r - 1} \cup X_{1, r},\bb_{1,r}))]
   &\geq (1 - \delta) \alpha \cdot \E_{X_{1,r}}[f(P'_r \cup \opt_{r-1})] - \eps f(\opt) \notag\\
  &\geq \alpha \cdot \E_{X_{1,r}}[f(P'_r \cup \opt_{r-1})] - (\eps + \alpha \delta) f(\opt) \notag\\
  &\geq (1 - \eps) \alpha \cdot f(P_r \cup \opt_{r-1}) - (\eps + \alpha \delta) f(\opt) \notag \\
  &\geq (1 - \eps) \alpha \cdot f(P_r \cup \opt_{r-1}) - 2\eps f(\opt). \label{eq:rand-greedy}
\end{align}
Here, the second inequality follows from the fact that $f$ is monotone restricted to $\opt \supseteq (P'_r \cup \opt_{r-1})$, the third from Lemma \ref{lem:lovasz-1} and the fact that every element of $P_r$ appears in $P'_r$ with probability at least $(1 - \eps)$, and the last from our assumption that $\alpha \delta \le \eps$.

Next, note that $Q'_r \subseteq (C_r \cap \opt) \setminus \opt_{r-1}$.  This together with monotonicity of $f$ restricted to $r$ imply:
\begin{align*}
  \Ex[f(C_r \cap \opt) \mid C_{r - 1} = \hC_{r - 1}] & \geq \Ex[f(Q'_r \cup \opt_{r-1})] \notag\\
  & \geq {1 \over 2} \cdot f(Q_r\cup \opt_{r-1}) + \frac{1}{2} \cdot f(\opt_{r-1}),
\end{align*}
where the last inequality follows from Lemma \ref{lem:lovasz-1} and Lemma \ref{lem:randomized-probabilities}.  Rearranging this inequality using the condition $C_{r-1} = \hC_{r-1}$ and the definition $\opt_{r-1} = \hC_{r-1} \cap \opt$ we obtain:
\begin{align}
  \Ex[f(C_r \cap \opt) - f(C_{r - 1} \cap \opt) \mid C_{r - 1} = \hC_{r - 1}]
  & \geq {1 \over 2} \left(f(Q_r \cup \opt_{r-1}) - f(\opt_{r-1})\right) \notag \\
  & \geq {1 \over 2} \left(f(P_r \cup Q_r \cup \opt_{r-1}) - f(P_r \cup \opt_{r-1}) \right)\notag\\
  & = {1 \over 2} \left(f(\opt) - f(P_r \cup (\hC_{r - 1} \cap \opt)) \right), \label{eq:rand-cr-increase}
\end{align}
where the second inequality follows from submodularity.

Now, if $f(P_r \cup (\hC_{r - 1} \cap \opt)) \geq (1 - \eps) \cdot f(\opt)$ then this fact together with (\ref{eq:rand-greedy}) imply the first property in the statement of Theorem~\ref{thm:randomized-parallel-main} must hold. Otherwise, $f(\opt) - f(P_r \cup (\hC_{r - 1} \cap \opt)) \geq \eps \cdot f(\opt)$; this fact together with (\ref{eq:rand-cr-increase}) imply that the second property must hold.
\end{proof}

\section{Analysis of \DCGreedy for the Application of the Randomized Framework}
\label{sec:dcgreedy-analysis}

\begin{figure}[t]
\begin{center}
\begin{minipage}{0.75\linewidth}
\begin{algorithm}[H]
\KwIn{$N \subseteq V$, $\bx \in [0, 1]^N$}
$W \gets \emptyset$\;
$\by \gets 0$\;
\InfLoop{}{
  Let $D \gets \{e \in N \setminus W \colon W \cup \{e\} \in \I$\}\;
  Pick $\ell = \Theta\left(\frac{s\log n}{\eps^2}\right)$ independent random samples for $R(\bx + \by)$\;
  \ForEach{$e \in D$}
  {
    $w_e \gets \textnormal{approximation of } \Ex[f(R(\bx + \by) \cup \set{e}) - f(R(\bx + \by))]$ via above samples\;
  }
  Let $e^* = \arg\max_{e \in D}\{w_e\}$\;
  \uIf{$D = \emptyset$ or $w_{e^*} < 0$}{
    \Return $\by$\;
   }
  \Else{$y_{e^*} \gets y_{e^*} + \eps(1 - x_{e^*})$\;
    $W \gets W \cup \{e\}$\;
  }
}
\caption{Greedy Update Step (\GreedyStep) with randomized approximation of $w_e$'s.}
\label{alg:rand-greedy}
\end{algorithm}
\end{minipage}
\end{center}
\caption{Discretized continuous greedy (\DCGreedy) with approximate evaluation of the $w_e$'s on line 7. The weight $w_e$ is estimated as follows. Given $\ell$ independent random sets $R_1, \dots, R_{\ell}$ (the samples for $R(\bx + \by)$), $w_e$ is set to ${1 \over \ell} \sum_{i = 1}^{\ell} (f(R_i \cup \set{e}) - f(R_i))$.}
\label{alg:dcgreedy-approx}
\end{figure}

In this section, we show that we can instantiate the randomized framework from Section~\ref{sec:parallel-framework-rand} with a modified \DCGreedy algorithm, and obtain the results stated in Section~\ref{sec:parallel-dcgreedy}. Specifically, we extend the \DCGreedy algorithm to the setting in which the weights $w_e$ on line $5$ of \GreedyStep are evaluated approximately via samples.  The resulting \GreedyStep is shown in Algorithm \ref{alg:rand-greedy}. 

We devote the rest of this section to proving the following result.

\begin{theorem} \label{lem:approx-greedy-properties}
  The modified \DCGreedy algorithm with approximate evaluation of $w_e$'s satisfies the consistency property and the $(\alpha, \eps, \delta)$-approximation property with $\delta = 1/n$ and $\alpha = 1/e - O(\eps)$ for non-monotone functions and $\alpha = 1 - 1/e - O(\eps)$ for monotone functions.
\end{theorem}

We begin by verifying that the consistency property holds.  Consider a vector $\bb$ and two subsets $A, B \subseteq V$ such that, for each element $e \in B$, we have $\DCGreedyRel(A \cup \set{e}, \bb) = \DCGreedyRel(A, \bb)$.    We shall show that the approximate \GreedyStep algorithm always picks the same set $W$ on input $(A, \bx, \bb)$ and $(A \cup B, \bx, \bb)$. (Note that, since the two runs have the same randomness $\bb$, they will use the same approximate weights.) Suppose for contradiction that the algorithm makes different choices on input $(A, \bx, \bb)$ and $(A \cup B, \bx, \bb)$. Consider the first iteration where the two runs differ, and let $e$ be the element added to $W$ in that iteration on input $(A \cup B, \bx, \bb)$. Note that $e \notin A$ and thus we have $e \in B$. But then $e$ would be added to $W$ on input $(A \cup \set{e}, \bx, \bb)$, as well. Thus $e \in \DCGreedyRel(A \cup \set{e}, \bb)$, which contradicts the fact that $e \in B$. Thus the consistency property holds.

Now we verify that the $(\alpha, \eps, \delta)$-approximation property holds. The analysis of the modified \DCGreedy algorithm is similar to the analyses in \cite{Feldman2011a,Badanidiyuru2014}.

\begin{lemma}\label{lem:dc-gen-approx}
Let $\I$ be matroid on $V$ and $\opt = \argmax_{S \in \I, S \subseteq V} f(S)$. Let $A \subseteq V$ and $\bb \sim \D$. Let $B \subseteq \opt$ be a subset such that, for each $e \in B$, $e \notin \DCGreedy(A \cup \set{e}, \bb)$. Then, we have
\[ 
\Pr_{\bb \sim \D}[F(\DCGreedy(A \cup B, \bb)) \geq \alpha \cdot f((A \cup B) \cap \opt) - \eps \cdot f(\opt)] \ge 1-1/n,
\]
where $\alpha = (1 - 1/e - O(\eps))$ for monotone $f$ and $(1/e - O(\eps))$ for general $f$.
\end{lemma}

In the remainder of this section, we prove Lemma \ref{lem:dc-gen-approx}. If $\bx, \by \in [0, 1]^N$, we denote by $\bx \lor \by$ the vector such that $(\bx \lor \by)_i = \max\set{\bx_i, \by_i}$. Similarly, $\bx \land \by$ is the vector such that $(\bx \land \by)_i = \min\set{\bx_i, \by_i}$.
 Let $\opt = \argmax_{S \subseteq V, S \in I} f(S)$ be an optimal solution over the entire ground set $V$, and consider the execution of $\DCGreedy(A \cup \opt, \bb)$.  Let $Z$ be the set of vectors that $\DCGreedy(A \cup \opt, \bb)$ considers when computing the weights of elements, i.e., the set of all vectors $\bz \coloneqq \bx + \by$, where $\bx = \bx(t)$ for some iteration $t$ of $\DCGreedy(A \cup \opt, \bb)$ and $\by$ is the vector constructed by previous iterations of $\GreedyStep(A \cup \opt, \bx, \bb)$.  Formally, we associate the vector $\bz_{j} \in Z$ with the $j$th execution of $\GreedyStep$'s main loop (counted across all the iterations of $\DCGreedy$).  Note that $|Z| \le s/\epsilon$, since \GreedyStep's loop is executed at most $s$ times for each of the $1/\epsilon$ iterations of $\DCGreedy$. For each sample, the random string $\bb$ can simply store $|V|$ random thresholds in $[0,1]$. For a given vector $\bz$, these thresholds can be used to round $\bz$ to an integral indicator vector (a sample of $R(\bz)$) in order to estimate $\Ex[f(R(\bz)\cup \set{e}) - f(R(\bz)]$.

Consider the $j$th time \GreedyStep executes line 7, and suppose that for each element $e \in A \cup \opt$ we compute a weight $w_e(\bz_j, \bb)$, by using $\ell$ random samples encoded by $\bb$ to estimate $R(\bz_j)$, as in \GreedyStep.  We say that $w_e(\bz_j, \bb)$ is a good estimate if
\[
		\card{w_e(\bz_j, \bb) - \Ex[f(R(\bz_j) \cup \set{e}) - f(R(\bz_j))]} \le {\eps \over 2s} f(\opt) + \frac{\eps}{2} \Ex[f(R(\bz_j) \cup \set{e}) - f(R(\bz_j))]. 
\]
		We say that $\bb$ is good if all of the weights $\set{w_e(\bz_j, \bb) \colon \bz_j \in Z, e \in A \cup \opt}$ are good estimates.  
\begin{lemma}
\label{lem:cg-randomness}
The randomness $\bb$ is good with probability at least $1 - 1/n$.
\end{lemma}
\begin{proof}
Let $d = \max_{e \in V} f(e) \leq f(\opt)$. Consider a weight $w_e$ and let $R_1, \dots, R_{\ell}$ denote the independent random sets used to compute $w_e$ in line 7 of \GreedyStep. For each $i \in [\ell]$, let $w_{e,i} = f(R_i \cup \set{e}) - f(R_i)$. Note that, by submodularity, $w_{e, i} \leq d \leq f(\opt)$. We use the following version of the Chernoff bound.

\begin{lemma}[Lemma~{2.3} in \cite{Badanidiyuru2014}]
\label{lem:chernoff}
  Let $X_1, \dots, X_m$ be independent random variables such that for each $i$, $X_i \in [0, 1]$. Let $X = {1 \over m} \sum_{i = 1}^m X_i$ and $\mu = \Ex[X]$. Then
  \begin{align*}
    \Pr[X > (1 + \alpha) \mu + \beta] &\leq \exp\left( - {m \alpha \beta \over 3} \right),\\
    \Pr[X < (1 - \alpha) \mu - \beta] &\leq \exp\left(- {m \alpha \beta \over 2} \right).
  \end{align*}
\end{lemma}
If we choose an appropriately large constant in the definition of $\ell$, then setting $m = \ell$, $X_i = w_{e, i} / d$, $\alpha = \eps/2$, and $\beta = \eps/2s$ in Lemma \ref{lem:chernoff}, we obtain that $w_e$ is a good estimate with probability at least $1 - 1 / n^4 \geq 1 - \eps/( s n^2)$. The size of $Z$ is at most $s/\eps$ and for each element of $Z$, there are at most $n$ weights to be estimated, so the lemma follows by the union bound.
\end{proof}

Now, note that if some random string $\bb$ is good, then all weights calculated by $\DCGreedy(A \cup B, \bb)$ in are good also, since $A \cup B \subseteq A \cup \opt$, and, as we have noted, the consistency property implies that \GreedyStep picks the same set on inputs $(A \cup B, \bx,\bb)$ and $(A,\bx,\bb)$ in each iteration.  We now fix some good $\bb$, and show that for \emph{any} $B \subseteq \opt \setminus A$, we must have:
\begin{equation}
\label{eq:CG-good-property-conditioned}
F(\DCGreedy(A \cup B,\bb)) \ge \alpha \cdot f((A \cup B) \cap \opt).
\end{equation}
Where $\alpha = 1/e - O(\eps)$ for non-monotone functions and $\alpha = 1 - 1/e - O(\eps)$ for monotone functions.  When $f$ is monotone, this follows from previous work \cite{Badanidiyuru2014}.  Thus we focus on the non-monotone case. This will finish the proof Lemma \ref{lem:dc-gen-approx}.

Let $N = A \cup B$ for some $B \subseteq \opt \setminus A$ and consider the restricted maximization problem $\max_{S \subseteq N, S \in \I}f(S)$.  We will need the following two lemmas from previous work. The first lemma is well-known and it follows from the exchange property of a matroid (see for example \cite{Schrijver2003}).

\begin{lemma}
\label{lem:matroid-exchange}
  Let $\mathcal{M} = (N, \I)$ be a matroid and let $B_1, B_2 \in \I$ be two bases in the matroid. There is a bijection $\pi: B_1 \rightarrow B_2$ such that for every element $e \in B_1$ we have $B_1 \setminus \set{e} \cup \set{\pi(e)} \in \I$. 
\end{lemma}

\begin{lemma}[\cite{Feldman2011a}]
\label{lem:multilinear-lb}
Consider a vector $\bx \in [0,1]^N$. Assuming $x_e \le a$ for every $e \in V$, then for every set $S \subseteq N$, $F(\bx \lor \one_{S}) \ge (1 - a)f(S)$.
\end{lemma}

Now, we begin by showing that \DCGreedy improves the current solution by a large amount in each step.

\begin{lemma}
\label{lem:cg-onestep}
  Suppose that the randomness $\bb$ is good. In each iteration $t$ of \DCGreedy, $F(\bx(t)) - F(\bx(t-1))\ge \eps (1 - \eps) ((1 - \eps)^tf(N \cap \opt) - F(\bx(t))) - \eps^2 f(\opt)$.
\end{lemma}
\begin{proof}
Fix an iteration $t$, and for brevity denote $\bx = \bx(t-1)$, $\bx' = \bx(t)$. Let $W$ be the set of elements selected by the \GreedyStep for this update, and let $\by$ be the associated update vector. We suppose without loss of generality that $|W| = s$, where $s$ is the rank of the matroid $\I$, since if $|W| < s$ we can simply add $s - |W|$ dummy elements to $W$. Let $e_i$ be the $i$th element added to $W$ by \GreedyStep and let $\by(i)$ be the value of $\by$ after $i$ elements have been added to $W$.

  By Lemma~\ref{lem:matroid-exchange}, there is a bijective mapping $\pi: N \cap \opt \rightarrow W'$ between $N \cap \opt$ and a subset $W' \subseteq W$ of size $\card{N \cap \opt}$ such that, for each element $o \in N \cap \opt$, $W \setminus \set{\pi(o)} \cup \set{o} \in \I$. For each $i \in [s]$, let $o_i \coloneqq \pi^{-1}(e_i)$ if $e_i \in W'$ and $o_i \coloneqq e_i$ otherwise.

  For each $i$, we have $w_{e_i} \geq w_{o_i}$, since $o_i$ is a candidate element during the iteration of \GreedyStep that picked $e_i$. Thus, since all the weights are good estimates, we have
  \begin{multline} \label{eq:rand-marginal}
    \Ex[f(R(\bx + \by(i-1)) \cup \set{e_i}) - f(R(\bx + \by(i-1)))] \\
\geq (1 - \eps) \Ex[f(R(\bx + \by(i-1)) \cup \set{o_i}) - f(R(\bx + \by(i-1)))] - {\eps \over s} f(\opt). 
  \end{multline}
for all $\by$ and $i$.
 Then, we have:
\begin{align}
F &(\bx') - F(\bx) = F(\bx + \by) - F(\bx) \notag\\
&= \sum_{i = 1}^s (F(\bx + \by(i)) - F(\bx + \by(i - 1)))  \notag\\
&= \sum_{i=1}^{s} \eps(1 - x_{e_i}) \frac{\partial F}{\partial x_{e_i}}\Big|_{\bx + \by(i - 1)} \notag\\
&= \sum_{i = 1}^{s}\eps\Ex[f(R(\bx + \by(i - 1)) \cup \set{e_i}) - f(R(\bx + \by(i - 1)))] \notag\\
&\ge \sum_{i = 1}^{s}\eps \left((1 - \eps) \Ex[f(R(\bx + \by(i - 1)) \cup \set{o_i}) - f(R(\bx + \by(i - 1)))] - {\eps \over s} f(\opt) \right) \notag \\
&\ge \sum_{i = 1}^{s} \eps\left((1 - \eps) \Ex[f(R(\bx') \cup \set{o_i}) - f(R(\bx'))] - {\eps \over s} f(\opt) \right)  \notag\\
&\ge \eps (1 - \eps)(F(\bx' \lor \one_{N \cap \opt}) - F(\bx')) - \eps^2 f(\opt),
\label{eq:rand-cg-main-lb}
\end{align}
where the first inequality follows from~\eqref{eq:rand-marginal} and the last two from submodularity.

We relate the value $F(\bx' \lor \one_{N \cap \opt})$ to $f(\opt)$ using Lemma~\ref{lem:multilinear-lb}. At each step, we increase each coordinate $e$ of $\bx$ by at most $\eps(1 - x_e(t))$.  Thus, for any step $0 \le j \le 1/\eps$, we have
  \[  x_e(j+1) - x_e(j) \le (1 - x_e(j))\eps,\]
or, equivalently,
  \[ x_e(j+1) - (1 - \eps)x_e(j) \le \eps. \]
Thus, for each time step $t \leq 1/\eps$, we have
\begin{equation*}
  x_e(t) \le \sum_{j = 1}^t \eps(1-\eps)^{t-1-j} = 1 - (1-\eps)^t
\end{equation*}
By combining the inequality above with Lemma~\ref{lem:multilinear-lb}, we obtain
\[ F(\bx(t) \lor \one_{N \cap \opt}) \ge (1-\eps)^{t}f(N \cap \opt).\]
Plugging this bound into \eqref{eq:rand-cg-main-lb} then completes the proof.
\end{proof}

\begin{lemma}
\label{lem:cg-allsteps}
  Suppose that the randomness $\bb$ is good. The final solution $\bx(1/\eps)$ constructed by $\DCGreedy(N, \bb)$ satisfies $F(\bx(1/\eps)) \ge (1/e - \eps) f(N \cap \opt) - \eps f(\opt)$. Therefore the integral solution $S$ satisfies $f(S) \geq (1/e - \eps) f(N \cap \opt) - \eps f(\opt)$.
\end{lemma}
\begin{proof}
By rearranging the inequality from Lemma \ref{lem:cg-onestep}, we obtain
\begin{align*}
  F(\bx(t)) &\ge \frac{\eps(1-\eps)^{t + 1} f(N \cap \opt) + F(\bx(t-1)) - \eps^2 f(\opt)}{1+\eps}\\
  &\ge \eps(1 - \eps)^{t+2}f(N \cap \opt) + (1 - \eps)F(\bx(t-1)) - \eps^2 f(\opt).
\end{align*}
It follows by induction that
  \[ F(\bx(t)) \ge t \eps (1 -\eps)^{t+2} f(N \cap \opt) - t \eps^2 f(\opt).\]
Thus
\begin{equation*}
F(\bx(1/\eps)) \ge (1 - \eps)^{\frac{1}{\eps} + 2}f(N \cap \opt) - \eps f(\opt) 
\ge \left({1 \over e} - \eps \right) f(N \cap \opt) - \eps f(\opt). \qedhere
\end{equation*}
\end{proof}
Combining Lemmas \ref{lem:cg-randomness} and \ref{lem:cg-allsteps} then complete the proof of Lemma \ref{lem:dc-gen-approx}.

\section{Improved Analysis of the Two-Round Algorithm for Non-monotone Maximization with a Cardinality Constraint}
\label{sec:two-round-cardinality}

In this section, we show that for a cardinality constraint, we can improve the analysis slightly of the algorithm given in Subsection~\ref{sec:two-round-algorithms}.

\begin{theorem}
\label{thm:two-round-cardinality}
If $\I$ is a cardinality constraint, the two-round algorithm from Subsection~\ref{sec:two-round-algorithms} achieves a $\left(1 - {1 \over m} \right) {1 - {1 \over e} \over 1 + {1 \over \beta} \left(1 - {1 \over e} \right)}$  approximation in expectation for non-monotone maximization, where $\beta$ is the approximation guarantee of \Alg. 
\end{theorem}
\begin{proof}
The analysis is similar to the one in the proof of Theorem~\ref{thm:two-round}, and we describe the main changes in this section. We define $O$ as before, and modify the analysis of the solution $S_1$ as follows.  Let $S_1^j$ be the subset of $S_1$ consisting of the first $j$ elements picked by \Greedy, with $S_1^0 = \emptyset$. By the standard analysis of the \Greedy algorithm for a cardinality constraint, for each $j \in [k]$, we have
  \[ f(S_1^j) - f(S_1^{j - 1}) \geq {f(S_1^{j - 1} \cup O) - f(S_1^{j - 1}) \over k}, \]
  and therefore
  \[ f(S_1) \geq \sum_{j = 0}^{k - 1} {1 \over k} \left(1 - {1 \over k} \right)^{k - 1 - j} f(S_1^j \cup O).\]
  Now, using $\Ex[\one_{S_1^j \cup O}] = \Ex[\one_{S_1^j}] + \Ex[\one_O] = \Ex[\one_{S_1^j}] + \one_\opt - \bp$, and Lemma \ref{lem:lovasz}, we obtain:
  \begin{equation*}
    \Ex[f(S_1^j \cup O)] \geq f^-(\Ex[\one_{S_1^j}] + \one_{\opt} - \bp).
  \end{equation*}
  Therefore
  \begin{equation}
    \Ex[f(S_1)] \geq \sum_{j = 0}^{k - 1} {1 \over k} \left(1 - {1 \over k} \right)^{k - 1 - j} f^-(\Ex[\one_{S_1^j}] + \one_{\opt} - \bp). \label{eq:first-machine-cardinality}
  \end{equation}
We analyze the expected value of the solution $T$ as before, obtaining \eqref{eq:last-machine} from page \pageref{eq:last-machine}. By combining (\ref{eq:first-machine-cardinality}) and (\ref{eq:last-machine}), we get
  \begin{align*}
    \Ex[f(S_1)] + {1 \over \beta} \left(1 - \left(1 - {1 \over k} \right)^k \right) \Ex[f(T)] &\geq \sum_{j = 0}^{k - 1} {1 \over k} \left(1 - {1 \over k} \right)^{k - 1 - j} \Big(f^-(\Ex[\one_{S_1^j}] + \one_{\opt} - \bp) + f^-(\bp) \Big)\\
   &\geq \sum_{j = 0}^{k - 1} {1 \over k} \left(1 - {1 \over k} \right)^{k - 1 - j} 2 \cdot f^-\left({\Ex[\one_{S_1^j}] + \one_{\opt} \over 2}\right),
  \end{align*}
where the last inequality follows from the convexity of $f^-$.  Since $S_1^j \subseteq V_1$ and $V_1$ is a $1/m$ sample of $V$, we have $\Ex[\one_{S_1^j}] \leq {1 \over m} \cdot \one_V$.  Therefore, using the definition of $f^-$ and the non-negativity of $f$, we obtain
   \[ 2 \cdot f^-\left({\Ex[\one_{S_1^j}] + \one_{\opt} \over 2}\right) \geq \left(1 - {1 \over m} \right) f(\opt).\]
  and thus
  \[ \Ex[f(S_1)] + {1 \over \beta} \left(1 - \left(1 - {1 \over k} \right)^k \right) \Ex[f(T)] \geq \left(1 - \left(1 - {1 \over k} \right)^k \right) \left(1 - {1 \over m} \right) f(\opt).\]
 	It follows that
\begin{align*} \max\set{\Ex[f(S_1)], \Ex[f(T)]} &\geq \left(1 - {1 \over m} \right) {\left(1 - \left(1 - {1 \over k} \right)^k \right) \over 1 + {1 \over \beta} \left(1 - \left(1 - {1 \over k} \right)^k \right)} f(\opt) 
\\ &\ge \left(1 - {1 \over m} \right) {1 - {1 \over e} \over 1 + {1 \over \beta} \left(1 - {1 \over e} \right)} f(\opt).
\end{align*}
\end{proof}

\end{document}